\documentclass[11pt]{article} 
\usepackage{setspace}
\usepackage{times}
\usepackage[top=1.0in, bottom=1.0in, left=1in, right=1in]{geometry}
\usepackage[usenames,dvipsnames]{color}
\usepackage{xcolor}
\usepackage[table]{xcolor} 
\usepackage{colortbl} 
\usepackage{array}    
\usepackage{tabularx}      
\usepackage[colorlinks=true,citecolor=blue,linkcolor=violet]{hyperref}
\usepackage{hyperref}
\usepackage{amsmath}
\usepackage{url}
\usepackage{bm} 
\usepackage{pdflscape}
\usepackage{bbm} 
\usepackage{amsbsy} 
\usepackage{tikz} 
\usepackage[nameinlink]{cleveref} 
\usepackage{mathtools}
\usepackage{tablefootnote} 
\usepackage{natbib}
\usepackage{amsfonts} 
\usepackage{amssymb}  
\usepackage{amsthm}   
\usepackage{wasysym} 
\usepackage{amsxtra}
\usepackage{bm}
\usepackage{enumitem}
\usepackage{multirow}
\usepackage{multicol}
\usepackage{caption, subcaption}
\usepackage{nomencl}
\usepackage{color,colortbl}
\usepackage{dsfont}
\usepackage{bold-extra}
\usepackage{fancybox}
\usepackage{fancyhdr}
\usepackage{graphics}
\usepackage{graphicx}
\usepackage{circuitikz} 
\usepackage{tikz}
\usetikzlibrary{patterns}
\usepackage{latexsym}
\usepackage{longtable}
\usepackage{wrapfig}
\usepackage{float}
\usepackage{comment}
\usepackage{algorithm,algorithmic} 

\usepackage{caption}
\newtheorem{theorem}{Theorem}
\newtheorem{lemma}[theorem]{Lemma}
\newtheorem{proposition}[theorem]{Proposition}

\newtheorem{remark}{Remark}
\newtheorem{example}{Example}

\usepackage[normalem]{ulem}

\newcommand{\myref}[2]{\hyperref[#2]{#1 \ref*{#2}}} 

\usepackage{array}
\newcolumntype{H}{@{}>{\iffalse}c<{\fi}@{}} 

\begin{document}
\title{Liquidation Dynamics in DeFi and the Role of Transaction Fees}
\author{Agathe Sadeghi\thanks{Stevens Institute of Technology, School of Business, Hoboken, NJ 07030. \textit{asadeghi@stevens.edu} Corresponding author.} \and Zachary Feinstein\thanks{Stevens Institute of Technology, School of Business, Hoboken, NJ 07030.}}
\date{}
\maketitle
\abstract{
Liquidation of collateral are the primary safeguard for solvency of lending protocols in decentralized finance. However, the mechanics of liquidations expose these protocols to predatory price manipulations and other forms of Maximal Extractable Value (MEV). In this paper, we characterize the optimal liquidation strategy, via a dynamic program, from the perspective of a profit-maximizing liquidator when the spot oracle is given by a Constant Product Market Maker (CPMM). We explicitly model Oracle Extractable Value (OEV) where liquidators manipulate the CPMM with sandwich attacks to trigger profitable liquidation events. We derive closed-form liquidation bounds and prove that CPMM transaction fees act as a critical security parameter. Crucially, we demonstrate that fees do not merely reduce attacker profits, but can make such manipulations unprofitable for an attacker. Our findings suggest that CPMM transaction fees serve a dual purpose: compensating liquidity providers and endogenously hardening CPMM oracles against manipulation without the latency of time-weighted averages or medianization.
}\\
\noindent\textbf{Keywords:} Lending Protocols, Automated Market Makers, Oracle Extractable Value, Liquidation events

\newpage
\section{Introduction}\label{sec:intro}
Decentralized Finance (DeFi) has grown exponentially since 2018, transforming the traditional financial landscape by offering blockchain-based financial services within the crypto ecosystem. One of the most significant developments in DeFi is the rise of borrowing and lending platforms (LPs), which allow crypto asset holders to earn interest or access additional funds by utilizing their existing assets. Since 2020, the outstanding debt in DeFi lending has increased dramatically, driven by the wider adoption of stablecoins and their role in providing stability to the volatile cryptocurrency market. As of this writing LPs hold around \$80 billion in total value locked (TVL).\footnote{\url{https://defillama.com/protocols/lending}, accessed October 2025}

The standard operational process within these platforms includes lenders, who have surplus funds, supplying assets which are often stablecoins, to a specific lending smart contract. The utilization ratio which is the proportion of supplied liquidity currently borrowed in relation to the total amount of funds available, determines the interest rate on lending platforms. Each lender receives a platform-specific utility token representing their deposit, often functioning as a certificate of deposit. These tokens accrue interest and have a value equivalent to the underlying asset. Depositors can withdraw their funds at any time, typically by redeeming their utility tokens back to the platform.

In exchange, borrowers provide risky crypto assets, like Bitcoin or Ether, as collateral to back loans, which are frequently denominated in stablecoins. Smart contracts act as the direct interface connecting depositors and borrowers, streamlining the entire loan lifecycle, which includes loan issuance, repayment tracking, and collateral liquidation. Furthermore, smart contracts are responsible for managing collateral by assigning a specific margin to each collateral type and rigorously enforcing overcollateralization requirements \citep{capponi23}.

\subsection{Motivation}\label{sec:intro-motivation}
Despite its innovations, DeFi lending is not immune to risks associated with market volatility, liquidity mismatches, and operational inefficiencies \citep{gudgeon20_1}. Additionally, the trade-off between risk and reward in DeFi lending can be analyzed empirically \citep{lehar22}. Volatile crypto prices frequently trigger loan liquidations when borrowers fail to maintain their Loan-to-Value (LTV) ratios.
LTV ratios quickly deteriorate in response to a drop in collateral prices, which leads to automated liquidations. Lending platforms run the risk of bad debt in the absence of prompt liquidation since declining collateral values might not be sufficient to cover the outstanding debt. These forced sales have the potential to further depress collateral prices, which could trigger a chain reaction of liquidations throughout the ecosystem. 
For example, during the January 2022 crypto sell-off, liquidations surged to their highest level since May 2021, erasing \$50 billion in borrowed asset value.\footnote{International Monetary Fund (2022). \textit{Global Financial Stability Report: Shock waves from the War in Ukraine Test the Financial System’s Resilience}, Chapter 3. Available at: \url{https://www.imf.org/-/media/Files/Publications/GFSR/2022/April/English/ch3.ashx}} DeFi platforms set their own maximum LTVs based on factors such as collateral type and market volatility.

Liquidation events are essential automated mechanisms in DeFi lending protocols that are intended to safeguard the system's solvency. The main cause of these occurrences is when a borrower's Health Factor (HF), calculated as the ratio of value of collateral to debt, drops bellow a specified liquidation threshold. While the LTV ratio calculates the borrowed value relative to the collateral, HF is the inverse relationship (adjusted for the liquidation threshold) indicating a position's proximity to liquidation. The protocols' solvency is undermined when the value of the collateral assets, which are usually risky cryptocurrencies falls drastically. The loan is now automatically considered eligible for liquidation since it is undercollateralized. This liquidation necessity demonstrates the importance of having a reliable price for the collateral asset available for the lending protocol contract.

These prices can be obtained from off-chain suppliers like Chainlink or from decentralized exchanges' on-chain oracles such as Uniswap. However, top platforms, like Aave and Compound rely on further robustness mechanisms like time-weighted average price (TWAP) feeds that update once per block, in order to reduce the risk of short-term manipulation of spot markets. Although these defenses increase the oracle's resistance to manipulation, they also create weaknesses in times of rapid crisis (like during the collapse of Terra/Luna) when off-chain oracles might malfunction or TWAPs might react too slowly to reflect sharp price drops. Using off-chain or time-averaged oracle mechanisms usually requires more conservative LTV thresholds to account for update latency, which ultimately lowers capital efficiency and raises the collateral requirements for borrowers.

Although it is theoretically possible for borrowers to avoid liquidation by adding more collateral or manually repaying their loans, this is frequently not feasible for regular users. Timely manual intervention is difficult due to the need for continuous monitoring of fluctuating collateral prices and transaction fee variability. Theoretically, any blockchain user in the DeFi ecosystem can keep an eye on HFs and start the liquidation process on these undercollateralized positions. But in reality, liquidation in DeFi has become a very specialized task that is primarily carried out by advanced bots. With the help of flash loans \citep{qin21} and other algorithms, these bots can quickly liquidate large amounts of undercollateralized positions, removing the need for capital. Over 70\% of liquidable positions are now instantly liquidated by these automated agents, according to empirical data, demonstrating the significant improvement in these liquidators' efficiency over time \citep{perez21}. The substantial potential profits, which can occasionally amount to hundreds of thousands of dollars from a single transaction, are what propel this high speed and scale and encourage ongoing advancements in liquidation methods.

We investigate the use of constant product market makers (CPMMs) as an on-chain, trustless and without a lag spot oracles. To assess their robustness, we explicitly model the Oracle Extractable Value (OEV) and Maximal Extractable Value (MEV) accessible to sophisticated liquidators who can manipulate prices within a single block. We then study how transaction fees imposed by AMMs can act as an inherent safeguard, enhancing intra-block spot oracles’ robustness to transient attacks: by decreasing the marginal profitability of price manipulation, AMM trading fees alter the incentives of liquidators and potential attackers. More precisely, when transaction fees are incorporated into liquidation and MEV-driven tactics such as sandwiching liquidation events, profitability declines and eventually becomes negative beyond a critical threshold fee. This highlights the dual function of transaction fees, often neglected in prior research, in compensating liquidity providers while strengthening oracle resilience by deterring intra-block price manipulation that might otherwise trigger unnecessary liquidations.
\subsection{Literature Review}\label{sec:intro-litrev}
DeFi lending platforms operate by receiving crypto assets as deposits and lending them out to borrowers who meet specific collateral requirements. These platforms rely on liquidity pools, which are pools of deposited assets made available for lending. Users who deposit their crypto assets into these pools earn interest in return. Aave, MakerDAO and Compound  are notable instances of decentralized lending protocols that have become very popular. Among DeFi lenders, Aave stands out for having the highest total value locked \citep{iftikhar25}.

For borrowers, DeFi platforms offer the ability to borrow crypto assets from these liquidity pools by posting their deposited assets as collateral. The interest rate on borrowing, denominated in the borrowed asset \citep{john23}, varies based on the utilization rate\footnote{The utilization rate of a crypto asset represents the proportion of total loans to total deposits for that asset on the platform. When the deposit pool has greater available liquidity, the lending rate tends to be lower.} of the lending pool, meaning higher demand for an asset can drive up borrowing costs. This mechanism helps balance supply and demand within the system. Collateral is central to managing risks in DeFi lending, as it safeguards the platform against market volatility. Borrowers are often required to overcollateralize their loans \citep{wang22}, meaning they must provide collateral worth more than the loan itself. This is achieved by applying a collateral factor which is a discount rate assigned to each asset type. For example, if an asset has a collateral factor of 0.8, a borrower can only borrow up to 80\% of the posted collateral's value. Some assets, like Tether (USDT) on certain platforms, have a collateral factor of zero, meaning they cannot be used as collateral. As noted by \cite{cornelli24}, the anonymity and volatility of crypto assets necessitate overcollateralization as the primary risk management tool, unlike traditional banking, which uses undercollateralized loans backed by diverse assets like real estate. This reliance on crypto assets as collateral also makes DeFi inherently self-referential.

Borrowers on DeFi platforms can repay their loans at any time, but they must continuously meet collateral requirements. If the value of their collateral falls below the required threshold due to adverse price movements, liquidators are incentivized to step in. Liquidators repay the borrower’s debt and acquire the collateral in exchange for a reward, known as the liquidation bonus. This mechanism ensures that the platform remains solvent even during market volatility. However, the liquidation process can be costly, particularly when collateral shortfalls arise during periods of high volatility, leaving the platform vulnerable to further destabilization if liquidations are delayed. 

DeFi lending also facilitates leveraged trading strategies, such as leveraged longs and short selling, which are widely used by investors. For instance, borrowers may post volatile crypto assets, such as Ethereum or Wrapped Bitcoin, as collateral to borrow stablecoins. Stablecoins are a type of cryptocurrency specifically created to maintain price stability \citep{clark20}. Their value is usually tied to a reference asset, such as the US dollar, to minimize fluctuations. For a thorough understanding of these coins we refer the reader to \cite{klages20}.

Many studies explore lending and borrowing in the context of DeFi. \cite{bartoletti21, bartoletti22} demonstrate essential behavioral characteristics of decentralized lending and borrowing protocols. \cite{perez21} presents an abstract framework to analyze the state of decentralized lending and borrowing protocols. \cite{szpruch24} explores loan contracts as American perpetual barrier options. \cite{mueller24} observes that the liquidation mechanism in DeFi lending protocols amplifies both the costs and risks associated with leveraged investments.

The liquidation mechanism should aim to minimize both the frequency of liquidation events and the total number of transactions, given the constrained nature of blockchains \citep{qin21_1}. \cite{qin23} proposes reversible call options, which allows the seller of a call option to cancel the contract before its maturity date. \cite{iragorri21} finds that if liquidators are incentivized sufficiently, the risk of undercollateralization will be kept low.

Bad debt can happen as well. This occurs when the value of the borrower’s collateral is insufficient to cover their outstanding debt (under-collateralizitation) after liquidation, exposing the protocol to financial losses \citep{gudgeon20}. The accumulation of bad debt diminishes the total liquidity available in a lending protocol, resulting in higher interest rates for borrowers. If the lending pool becomes entirely comprised of bad debt, lenders will be unable to withdraw their funds. Another risk to consider is that the liquidator's transaction fee may exceed the value of the discounted collateral obtained, failing to provide an incentive for them to assist in closing the position.

Peer-to-peer and peer-to-pool marketplaces known as Decentralized Exchanges (DEXs) facilitate direct transactions between cryptocurrency traders. A key potential of decentralized finance is embodied by this architecture: enabling financial transactions without the need for conventional middlemen like banks, brokers, or payment processors.

There have been two main developments in DEX architectures. Launched in 2016, Order-Book-Based DEXs are the first version of these exchanges \citep{daian20}. By using a central limit order book where buyers and sellers place bid and ask orders, which are subsequently matched to execute trades, these platforms function similarly to traditional stock exchanges. With the advent of AMMs in 2018, a major architectural breakthrough was made possible by protocols such as Uniswap \citep{angeris20}. The AMM model uses liquidity pools, which usually consist of two separate asset reserves, to facilitate trading rather than an order book. An AMM's core implementation is made possible by a number of smart contracts that run on a blockchain. In order to determine prices algorithmically, pioneering AMMs like Uniswap v2 use a constant product function, which can be mathematically expressed as $x \times y = k$, where $x$ and $y$ are the quantities of two tokens in the pool and $k$ is their constant product \citep{angeris21}. This mathematical connection guarantees that the current supply and demand dynamics within that particular pool will determine the token price.

Transaction and gas fees are important aspects of user behavior and market dynamics in DEXs as well. On blockchain platforms like Ethereum, gas fees, which are based on network congestion and the computational complexity of transactions, act as incentives for validators to add transactions to blocks \citep{cong21}. DEX protocols may impose transaction fees in addition to gas in order to fund liquidity providers or platform upkeep \citep{capponi21}. These expenses can have a big influence on trading tactics, especially in situations with high frequency or low margins where fee optimization is crucial. The efficiency and accessibility of decentralized trading are shaped by the interaction of gas prices, transaction throughput, and user incentives. This interaction affects market participation as well as the general scalability of DEX ecosystems.

Volatility in collateral prices due to liquidations can exacerbate price impact in the DEX, creating feedback loops. That is why price impact matters. Price impact refers to the effect of asset liquidation on the market price of those assets, particularly in times of financial distress. \cite{zach17} introduces a framework that models the dynamics of financial contagion caused by price impact and fire sales in systems with multiple illiquid assets in traditional financial markets.

Maximal Extractable Value (MEV) is the extra profit that block producers, like miners, can make by choosing the order in which transactions go into a block. This matters a lot in areas like DeFi, where the order of transactions can affect the outcome \citep{daian20}. This ability to reorder transactions allows for practices akin to frontrunning \citep{eskandari19} and sandwich trading \citep{zhou21}, which would be considered illegal in traditional financial markets \citep{auer22}. Specialized actors known as searchers track pending transactions to spot profit-making opportunities, like arbitrage between exchanges, atomic arbitrage or liquidations on lending platforms. They then submit their own carefully crafted transactions or bundles to capture that value. Block producers can include these transactions in the blocks they create, often giving priority to those offering the highest fees or direct payments. By controlling the order and composition of transactions, they can extract additional value from the block \citep{qin22,li23}. MEV also leads to inefficiencies by filling blocks with profit-seeking transactions and pushing out valuable user activity \citep{capponi25}. Some work has been published on how to mitigate the MEV effect \citep{zhou21-1,heimach22,yang23}. There are also other kind of attacks that can be done in DeFi, see e.g., \cite{lee23}.

One specific kind of attack involves deliberately pushing the price of the risky asset down by selling large amount of it first. This is putting a sell order at the beginning of the block, i.e., frontrunning. Followed with other orders, the attacker puts a buy order at the end of the same block, allowing them to repurchase the asset with the lowered price, i.e., backrunning. They can potentially profit from the classic selling high, buying low strategy, known as sandwiching in the DeFi space. When the initial frontrunning is used to manipulate the price of an asset for use in another smart contract, the profits from this attack can be characterized as OEV. We will explicitly consider attacks of this type within this work.
\subsection{Primary Contributions}\label{sec:intro-contribs}
\paragraph{Optimal Liquidations} In contrast to models focused on the lending platform’s perspective and the mitigation of bad debt (see, e.g., \cite{qin21_1}), we examine liquidation spirals from the liquidator's point of view in Section~\ref{sec:optimiz-single}. In order to isolate how real-time price impact shapes liquidation incentives, particularly during fast-moving crises, we assume that the lending platform quotes the intra-block spot price to measure the health of the loan's collateral. A dynamic program for the liquidator's profit maximization within a single block along with a semi-analytic characterization of the value function and an implementable algorithm are presented.

\paragraph{Switching Condition} We monitor three practical limits on how far a troubled position can be pushed by a liquidator. First, the maximum amount of collateral that the liquidator can confiscate (after deducting the discount and fees) is known as the collateral cap. Second, the debt cap, is the amount of liquidation that completely settles the borrower's outstanding debt. If the liquidator reaches the debt cap, even if most or all of the collateral was taken, the loan is closed and there is no bad debt. Third, the closing-factor cap, is not a hard policy limit; it’s the HF checkpoint. It is the liquidation amount at which the position's health factor reaches the pertinent threshold. These three make it simple to interpret the results: if we can either (a) reach the debt cap and retire the loan, or (b) reach the closing-factor cap and return the position's health to 1 before depleting collateral, the liquidation is successful. Bad debt only appears when we run out of collateral or the HF is restored to 1 before the debt is paid off in full. That is, as much collateral has been liquidated as possible but there is still a loan balance remaining to be paid.

\paragraph{Impact of Fees on MEV Profitability} We show how incorporating AMM trading fees into the liquidation and any MEV triggered (sandwich) attacks, compress the profitable region and can even cause reversion for bundles backed by flash loans in Section~\ref{sec:front}. We show that smaller, staged liquidations outperform lump-sum ones when fees are present. We numerically determine the fee levels at which all viable attacks lose their profitability in Section~\ref{sec:front-limit}, emphasizing fees as a design lever for strengthening AMM-based oracles in Section~\ref{sec:front-fees}.
\section{Background}\label{sec:background}
In this section, we discuss the stylized modeling framework for the lending and exchange protocols which we focus on within this work. In order to maintain consistency and clarity, we also introduce the notation that will be used throughout the paper.
\subsection{Lending Protocols}\label{sec:background-lending}
The health factor is a key statistic in DeFi lending platforms, measuring the stability of a borrower's position relative to the risk of liquidation. It represents the ratio of the liquidation-threshold-weighted collateral value to the total debt value, both denominated in a constant num\'{e}raire (such as ETH or USDC). The weights are determined by the collateral-specific liquidation thresholds (also known as the haircut rate), which reflect the maximum borrowing capacity of each asset before liquidation is triggered. The general formula for the health factor is given by
$$HF= \frac{\text{liquidation threshold}\times \text{amount of collateral} \times \text{price of collateral}}{\text{value of borrowed assets}}.
$$
The numerator sums the weighted values of all collateral tokens deposited by the user, adjusted by their respective liquidation thresholds. The denominator accounts for the user’s total debt. All values are denominated in, e.g., USDC to maintain consistency across assets. A higher HF reflects a more secure position, whereas a value below 1 indicates that the user’s position is undercollateralized and at risk of liquidation. Platforms, in practice, enforce two thresholds: if the HF falls below the Closing Factor (CF), then up to 100\% of the debt can be liquidated to fully resolve the position; if the HF falls below 1 but still higher than CF, only a portion ($\kappa$) of the borrower's debt becomes liquidatable.

Figure~\ref{fig:hf_schem} depicts the platform’s balance sheet before and after a liquidation. The platform holds $c$ units of collateral priced at $p_0$ against debt $b$ prior to liquidation; only a portion $\theta$ of the collateral's market value counts toward safety, so HF compares $\theta\,p_0\,c$ to $b$. A liquidation is allowed when HF is less than $1$ or the closing-factor threshold $CF$. In that step, a liquidator receives $x\,(1+\ell)$ units of collateral and repays $\beta(x)$ of the debt, where $x$ is the liquidation size and $\ell$ is the liquidation bonus. After liquidation, the remaining debt is $b-\beta(x)$, and the platform's held collateral becomes $c-x\,(1+\ell)$, marked at the updated price $p_1$. The collateral transferred out, $x\,(1+\ell)$, and the reduced debt, $\beta(x)$, are represented by the greyed out box at the top of the figure. After the liquidation, HF is defined as the same ratio, eligible collateral value over debt value, but calculated using the updated parameters. The position is healthy if the post-liquidation HF goes back to at least $1$. Otherwise, further liquidations might be necessary, and in the worst case, bad debt might remain. The formula for HF after the first liquidation is:
\begin{align}
HF&=\frac{\theta\,(c-x\,(1+\ell))\,p_1}{b-\beta(x)}\label{eq:hf_schem}
\end{align}

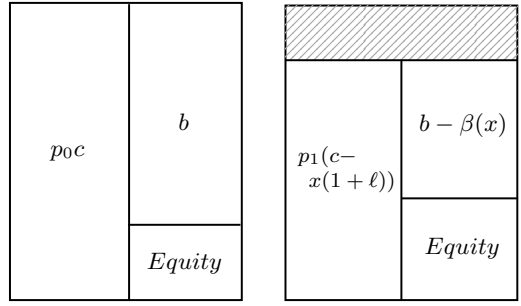
\begin{figure}[!tb]
\centering
\tikzset{every picture/.style={line width=0.75pt}} 

\begin{tikzpicture}[x=0.75pt,y=0.75pt,yscale=-1,xscale=1]

\draw   (101.75,79.67) -- (218.67,79.67) -- (218.67,229.67) -- (101.75,229.67) -- cycle ;
\draw    (161.83,80.33) -- (161.83,229.67) ;
\draw    (161.33,191.67) -- (218,191.67) ;
\draw   (240.71,81) -- (357.63,81) -- (357.63,229.67) -- (240.71,229.67) -- cycle ;
\path[pattern=north east lines, pattern color=gray!70]
  (240.71,81) rectangle (357.63,108.6);
\draw    (299.17,108.6) -- (299.17,230) ;
\draw    (299.33,178.33) -- (358,178.33) ;
\draw    (240.47,108.6) -- (357.87,108.6) ;

\draw (120.67,150.07) node [anchor=north west][inner sep=0.75pt]  [font=\footnotesize]  {$p_{0} c$};
\draw (169.33,204.07) node [anchor=north west][inner sep=0.75pt]  [font=\footnotesize]  {${\displaystyle Equity}$};
\draw (309.33,196.73) node [anchor=north west][inner sep=0.75pt]  [font=\footnotesize]  {${\displaystyle Equity}$};
\draw (185.33,134.07) node [anchor=north west][inner sep=0.75pt]  [font=\footnotesize]  {$b$};
\draw (306.67,134.07) node [anchor=north west][inner sep=0.75pt]  [font=\footnotesize]  {$b-\beta ( x)$};
\draw (239,149.38) node [anchor=north west][inner sep=0.75pt]  [font=\scriptsize,rotate=-0.03]  {$ \begin{array}{l}
p_{1}( c-\\
\ \ x( 1+\ell ))
\end{array}$};
\end{tikzpicture}
\caption{\footnotesize{Lending platform reserves \textbf{Left} initial state \textbf{Right} after liquidation state}}
\label{fig:hf_schem}
\end{figure}

Without timely liquidation, collateral shortfalls can undermine platform solvency. Overcollateralization mitigates some risks, but expected losses still average around 0.9\%, with riskier borrowers incurring larger losses as noted in \cite{imf22}. These dynamics underscore the importance of managing collateral and liquidation risks to ensure the resilience of DeFi platforms in volatile market conditions.
\subsection{Decentralized Exchanges}\label{sec:background-dex}
Let $A_0$ and $B_0$ denote the initial AMM reserves of a risky asset (e.g., collateral) and a stablecoin (e.g., USDC), respectively. The AMM charges a transaction fee $\gamma$ (e.g., 30 basis points in Uniswap v2) on the incoming asset. In practice, this means only $1-\gamma$ of the asset is being deposited in the pool, while liquidity providers receive the remaining amount as payment for their capital contributions. The constant product market maker (CPMM) rule in Uniswap v2 incorporates this fee structure directly, thereby modifying the reserves after each swap. Given this rule, the initial spot price at the CPMM is $p_0 = \frac{B_0}{A_0}$.
Under the scenario where we swap $a$ units of asset $A$ for $b$ units of asset $B$, the invariant must hold:
\begin{align*}
(A_0+a (1-\gamma))(B_0-b)=A_0B_0\, \quad \Rightarrow \quad 
b=\frac{B_0a (1-\gamma)}{A_0+a (1-\gamma)}
\end{align*}
Hence, after the swap, the updated reserves of the assets held by the CPMM will be $A_1=A_0+a (1-\gamma)$ and $B_1=B_0-b=\frac{A_0 B_0}{A_0 +a (1-\gamma)}$ while the updated price becomes $p_1=\frac{B_1}{A_1}$. Figure~\ref{fig:hyperbola} visualizes the reserve update along the constant product hyperbola when neglecting the fees ($\gamma = 0$).
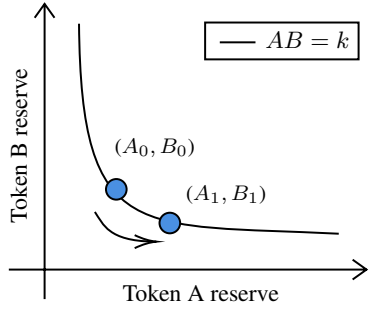
\begin{figure}[!tb]
\centering
\tikzset{every picture/.style={line width=0.75pt}} 

\begin{tikzpicture}[x=0.75pt,y=0.75pt,yscale=-1,xscale=1]

\draw  (50,227.7) -- (230.4,227.7)(68.04,93.6) -- (68.04,242.6) (223.4,222.7) -- (230.4,227.7) -- (223.4,232.7) (63.04,100.6) -- (68.04,93.6) -- (73.04,100.6)  ;
\draw    (85.4,103.6) .. controls (86.4,215.6) and (115.4,204.6) .. (216.4,209.6) ;
\draw  [fill={rgb, 255:red, 74; green, 144; blue, 226 }  ,fill opacity=1 ] (99,187.2) .. controls (99,184.33) and (101.33,182) .. (104.2,182) .. controls (107.07,182) and (109.4,184.33) .. (109.4,187.2) .. controls (109.4,190.07) and (107.07,192.4) .. (104.2,192.4) .. controls (101.33,192.4) and (99,190.07) .. (99,187.2) -- cycle ;
\draw  [fill={rgb, 255:red, 74; green, 144; blue, 226 }  ,fill opacity=1 ] (126,204.2) .. controls (126,201.33) and (128.33,199) .. (131.2,199) .. controls (134.07,199) and (136.4,201.33) .. (136.4,204.2) .. controls (136.4,207.07) and (134.07,209.4) .. (131.2,209.4) .. controls (128.33,209.4) and (126,207.07) .. (126,204.2) -- cycle ;
\draw    (93.4,198.6) .. controls (97.69,205.03) and (100.27,213.18) .. (122.64,213.58) ;
\draw [shift={(124.4,213.6)}, rotate = 180] [color={rgb, 255:red, 0; green, 0; blue, 0 }  ][line width=0.75]    (10.93,-3.29) .. controls (6.95,-1.4) and (3.31,-0.3) .. (0,0) .. controls (3.31,0.3) and (6.95,1.4) .. (10.93,3.29)   ;
\draw    (156.4,113) -- (174.4,113) ;
\draw   (149,103) -- (226.4,103) -- (226.4,121.6) -- (149,121.6) -- cycle ;

\draw (49.32,205.46) node [anchor=north west][inner sep=0.75pt]  [rotate=-270.24] [align=left] {{\footnotesize Token B reserve}};
\draw (106,234) node [anchor=north west][inner sep=0.75pt]   [align=left] {{\footnotesize Token A reserve}};
\draw (102,160) node [anchor=north west][inner sep=0.75pt]   [align=left] {{\scriptsize ($A_0, B_0$)}};
\draw (138,184) node [anchor=north west][inner sep=0.75pt]   [align=left] {{\scriptsize ($A_1, B_1$)}};
\draw (178,106) node [anchor=north west][inner sep=0.75pt]   [align=left] {{\footnotesize $AB=k$ }};
\end{tikzpicture}
\caption{\footnotesize{Constant product hyperbola. $k$ is a constant.}}
\label{fig:hyperbola}
\end{figure}

\section{Optimal Liquidation at Lending Protocols}\label{sec:optimiz}
The liquidation process in the lending platform creates an optimization problem: the searcher (same as liquidator in our case) aims to maximize liquidation profit under conditions such as (i) the HF must fall to or below 1 for liquidation to be valid, (ii) there must be sufficient available collateral and outstanding loan balance, and (iii) the prevailing market price updates based on the holdings of a CPMM. We interpret the resulting profits as an upper bound, as we ignore gas costs.

We specifically model this problem from the viewpoint of the liquidator, not the lending protocol, emphasizing the motivations and tactics of the liquidator. Two main strategies emerge for the searcher: Full liquidation in a single block to realize a one-time profit or a gradual liquidation, keeping HF just below 1 (or the closing factor) to repeatedly liquidate over multiple steps, potentially extracting more value.\footnote{Later, in Lemma~\ref{prop:dp}, we show that small liquidations are more profitable than a lump sum liquidation.} We use dynamic programming to model the strategic behavior of the liquidator, since the optimal action depends on expected future profits which in turn depend on future asset prices.

Additionally, we assume that the spot price from the CPMM is used by the lending protocol to calculate the HF. While most production platforms use TWAPs to reduce short-term volatility, using the instantaneous spot price gives a clearer picture of how intra-block price fluctuations directly affect liquidation incentives. Since the spot price is the primary driver from which TWAPs are ultimately derived, this modeling decision enables us to isolate and analyze the dynamic interaction between price impact and liquidation behavior without loss of generality which is the paper's main focus.

We also recall that the amount of a borrower's debt that can be liquidated is determined by two parameters; the health factor threshold CF, which determines when a position becomes eligible for liquidation and liquidation fraction $\kappa$, which specifies the fixed portion of the outstanding debt that can be repaid in a single liquidation transaction. We first analyze the case under fixed parameters, i.e., $(CF,\kappa)$ and afterwards generalize it to more flexible cases where several threshold pairs, i.e.\ $(CF_i,\kappa_i)$, are taken into account in order to identify the combination that maximizes profits and capture various liquidation regimes.

In this work, we extend the approach of, e.g., \cite{cohen23}, to incorporate transaction fees and examine scenarios where a liquidator can initiate liquidation by crashing the collateral price on DEXs (Section~\ref{sec:front}). We formalize this setup below, focusing on the mechanics of Aave\footnote{Aave is a DeFi lending platform built on the Ethereum blockchain, designed to facilitate borrowing and lending of cryptocurrency assets.} as the lending platform and Uniswap v2 as the DEX where the liquidity is distributed evenly across ticks.

It is helpful to introduce the shorthand $HF(c,b,A,B) := \frac{\theta B\,c}{A\,b}$ which represents the borrower's health factor in any given state. In this case, the \emph{updated} quantities corresponding to that state are $(c,b,A,B)$. In other words, following a liquidation trade of size $x$, the updated collateral becomes $c - x(1+\ell)$, the updated debt becomes $b - \beta(x)$ and the updated AMM reserves produce the updated price $p_1$. While the schematic expression in~\eqref{eq:hf_schem} showed how these variables change, the compact form above assesses the health factor following the updates. The corresponding health factor can therefore be obtained by directly substituting the current values of $(c,b,A,B)$ at each stage of the process.
\subsection{Single Threshold Health Factor}\label{sec:optimiz-single}
Following the notation introduced in Section~\ref{sec:background}, we now consider objective of a liquidator at the lending protocol. 
That is, finding the ideal liquidation size $x$ in a sequenced transaction in order to maximize the overall anticipated profit. Every liquidation generates an instant profit and, if the resulting health factor stays below the closing threshold $CF$, the liquidator can continue to make money from subsequent liquidations. In this way, given the current collateral $c$, outstanding debt $b$, and DEX reserves $(A,B)$, the ideal liquidation strategy is captured by the value function $V(\cdot)$ of the dynamic program:
\begin{align}
\label{eq:V} V(c,b,A,B) &= \begin{cases} \underset{\substack{x \geq 0 \\ x \leq x_c \\ x \leq x_{\kappa b} \\ x \leq x_{cf}}}{\sup} \Bigg[\begin{array}{l} \frac{Bx(1+\ell)(1-\gamma)}{A+x(1+\ell)(1-\gamma)}-\frac{Bx}{A} \\ + V(c-x(1+\ell),b-Bx/A,\bar{A},\bar{B})\end{array}\Bigg] & \text{if } HF(c,b,A,B) \leq CF \\ 0 & \text{if } HF(c,b,A,B) > CF  \end{cases}\\
\nonumber \bar{A} &= A+x(1+\ell)(1-\gamma), \\
\nonumber \bar{B} &= \frac{AB}{A+x(1+\ell)(1-\gamma)}.
\end{align} 
The liquidator selects the liquidation size $x$ at each stage within the feasible range of $0$ to $\min\{x_c,x_{\kappa b}, x_{cf}\}$, where $x_c = \frac{c}{1+\ell}$ and $x_{\kappa b} = \frac{\kappa bA}{B-\kappa b(1-\gamma)(1+\ell)}$ represent the amount to be liquidated to clear all collateral and debt (subject to the $\kappa$ constraint), respectively and $x_{cf}$ is the amount to be liquidated to revive the $HF=CF$. The liquidation profit $\frac{Bx(1+\ell)(1-\gamma)}{A+x(1+\ell)(1-\gamma)}-\frac{Bx}{A}$ is the immediate payout from performing the liquidation at the DEX where $\frac{Bx}{A}$ is paid to the lending protocol to claim the collateral. The process proceeds recursively, updating the state variables in accordance with the constant product pricing rule and accounting for transaction fees $\gamma$ and discount rate $\ell$, if the health factor following liquidation, $HF$, stays below the closing factor $CF$.

\begin{remark}
During liquidation, the health factor in~\eqref{eq:hf_schem} updates to reflect changes in collateral, debt, and price. After swapping collateral into the debt asset against a CPMM with reserves $(A,B)$, the post swap price is
$$p_1=\frac{B-y}{A+x(1-\gamma)(1+\ell)}=\frac{BA}{\bigl(A+x(1-\gamma)(1+\ell)\bigr)^2}\,,$$
where
\begin{align}
    y=\frac{B\,x(1+\ell)(1-\gamma)}{A+x(1+\ell)(1-\gamma)} \label{eq:y}
\end{align}
is the amount of debt asset received through the swap.
Substituting into the updated health factor yields
\begin{equation}
HF
=
\frac{\theta (c-x(1+\ell))BA}
{\kappa b\bigl(A+x(1-\gamma)(1+\ell)\bigr)^2
 -Bx\bigl(A+x(1-\gamma)(1+\ell)\bigr)}.
\label{eq:hf_fee_remark}
\end{equation}
Since $HF>0$, the liquidator's trade size $x$ must satisfy two feasibility constraints:
\begin{align}
c-x(1+\ell)\ge 0
\quad&\leftrightarrow\quad
x \le x_c := \frac{c}{1+\ell}
\label{eq:xc}\\
\kappa b-\frac{Bx}{A+x(1-\gamma)(1+\ell)}\ge 0
\quad&\leftrightarrow\quad
x \le x_{\kappa b} := \frac{\kappa bA}{B-\kappa b(1-\gamma)(1+\ell)}
\label{eq:xb}
\end{align}
The bound $x_c$ prevents the liquidator from withdrawing more collateral than the borrower holds, while $x_{\kappa b}$ prevents repaying more than the outstanding debt (subject to the maximum liquidation size $\kappa$). Throughout this work, in addition, we often consider the bound $x_b$ for repaying the full debt; this value results from taking $\kappa = 1$ in the above equation for $x_{\kappa b}$.
Often, for tractability, we impose the additional condition
$$b \le \frac{B}{(1-\gamma)(1+\ell)}\,,$$
in order to ensure the AMM has sufficient debt-asset liquidity relative to the position being unwound.
Finally, the trade size at which the liquidation threshold is met, $HF=CF$, is obtained by solving
$$\frac{\theta (c-x(1+\ell))BA}{b\bigl(A+x(1-\gamma)(1+\ell)\bigr)^2-Bx\bigl(A+x(1-\gamma(1+\ell)\bigr)}=CF\,,$$
which gives
\begin{equation}
x_{cf}=\frac{\Lambda + \sqrt{D}}{2CF\cdot\Gamma}\,,
\label{eq:x_cf}
\end{equation}
where $\Lambda=CF\cdot\bigl(2Ab(1-\gamma)(1+\ell)-BA\bigr)+\theta BA(1+\ell)$, $\Gamma=B(1-\gamma)(1+\ell) - b\,(1-\gamma)^2(1+\ell)^2$
and $D=\Lambda^2+4CF\cdot\Gamma\bigl(CF\cdot bA^2-\theta BAc\bigr)$.\\
In the optimization problem~\eqref{eq:V}, subject to $CF \geq HF(c,b,A,B)$, the admissible search domain for $x$ is therefore
$$0 \le x \le \min\{x_c,\;x_{\kappa b},\;x_{cf}\}\,,$$ where the liquidation thresholds are taken as infinity if they involve division by $0$.
\end{remark}

\begin{lemma} \label{prop:dp}
Consider a lending protocol holding collateral $c$ and debt $b$ with closing factor $CF$ and liquidation threshold $\kappa$. 
Consider, also, a CPMM with initial reserves $(A,B)$ in the collateral and debt assets, respectively.
The optimal profit~\eqref{eq:V} obtained by the liquidator is given as the output of Algorithm~\ref{algo:pi_liq}.
This profit is obtained by incrementally claiming and selling marginal $dx$ units of the collateral until the health factor equals $CF$ and then completing a final, finite sized, liquidation.
\end{lemma}
\begin{proof}
Herein we provide a sketch of the logic. For the complete proof refer to Appendix~\ref{app:dp}. 
First, we find that it is optimal for the liquidator to make sequential small liquidations over a single large liquidation. In this way the optimal strategy is to liquidate marginal $dx$ units of the collateral for as long as possible. Notably, this occurs up to whichever of the following conditions is met first: the collateral has been fully used ($x_c$), the loan is fully repaid ($x_b$), or the health factor has recovered up to $CF$ ($x_{cf}$). 
Because such liquidations occur at marginal size, we are able to find the total profits from this stage via the integral
\begin{equation}
    \pi_{liq}=\int_0^{x_{liq}} \frac{(B-y)((1-\gamma)(1+\ell)-1)}{A+x(1-\gamma)(1+\ell)}dx = \frac{B((1-\gamma)(1+\ell)-1)x_{liq} }{A+\,x_{liq}(1-\gamma)(1+\ell)} \label{eq:pi_liq}
\end{equation}
At this time, if further liquidations can be accomplished (i.e., $x_{liq} = x_{cf}$), the liquidator makes one final large transaction chosen to either exhaust the remaining collateral ($x_c - x_{cf}$), repay the maximum allowable fraction of the loan ($x_{\kappa b}$), or maximize total profitability ($x^*$).
\end{proof}
\begin{algorithm}[t]
\caption{\textbf{Liquidation Profit Dynamics $L(CF, \kappa)$}}
\label{algo:pi_liq}
\textbf{Input:} $A$ (collateral amount in DEX), $B$ (debt amount in DEX), $\gamma$ (transaction fee), $b$ (debt in LP), $c$ (collateral in LP), $\theta$ (haircut rate), $\ell$ (discount rate), $CF$ (closing factor), $\kappa$ (max liquidation rate) \\
0. Compute: $HF(c,b,A,B)=\frac{\theta B c}{A b}$, \textbf{If} $HF>CF$ or $\gamma\leq\frac{\ell}{1+\ell}$ \textbf{then Return:} $\pi_{\text{tot}} = 0$\\
1. Compute: $x_c = \frac{c}{1+\ell}, \quad x_b = \frac{bA}{B - b(1-\gamma)(1+\ell)}$\\
2. Compute: $x_{cf} = 
\frac{\Lambda \,+\, \sqrt{D}}{2CF\cdot\,\Gamma},$
\begin{flushleft}
$\left|\begin{array}{l}
\Lambda = -CF\cdot\bigl(2Ab\,(1-\gamma)(1+\ell) - BA\bigr)+ \theta BA(1+\ell)\\
\Gamma = B\,(1-\gamma)(1+\ell) - b\,(1-\gamma)^2(1+\ell)^2\\
D = \Lambda^2 + 4CF\cdot\,\Gamma\bigl(CF\cdot\,bA^2 - \theta BA c\bigr)\\ 
\end{array}\right.$
\end{flushleft}
3. Set: $x_{\text{liq}} = \min\{x_c, x_b, x_{cf}\}$ \\
4. Compute: $\pi_{\text{liq}} = \frac{B((1-\gamma)(1+\ell)-1)x_{liq} }{A+ x_{liq}(1-\gamma)(1+\ell)}$\\ 
5. \textbf{If} $x_{cf} < \min\{x_c, x_b\}$:
\begin{flushleft}
$\left|\begin{array}{l}
y_{cf} =\frac{Bx_{cf}(1-\gamma)(1+\ell)}{A+x_{cf}(1-\gamma)(1+\ell)} \\
\bar{b} = b - \frac{1}{1+l} y_{cf}\\
\bar{A} = A + x_{cf}(1-\gamma)(1+\ell),\quad \bar{B} = B - y_{cf}\\
x_{\kappa b} = \frac{\kappa \bar{b}\bar{A}}{\bar{B} - \kappa \bar{b}(1-\gamma)(1+\ell)}\\
x^*=\frac{\bar{A}(\sqrt{(1-\gamma)(1+\ell)} - 1)}{(1-\gamma)(1+\ell)}\\ x_{\text{last}} = \min\{x_c-x_{cf}, x_{\kappa b}, x^*\}\\
\pi_{\text{last}} = \frac{\bar{B}x_{last}(1+\ell)(1-\gamma)}{\bar{A}+x_{last}(1-\gamma)(1+\ell)}-\frac{\bar{B}x_{last}}{\bar{A}}
\end{array}\right.$
\end{flushleft}
6. \textbf{Else:} $\pi_{\text{last}} = 0$ \\
7. \textbf{Return:} $\pi_{\text{tot}} = \pi_{\text{liq}} + \pi_{\text{last}}$
\end{algorithm}

\begin{remark}
When calculating the profit with transaction fees at the CPMM, it is possible that the liquidator is taking a loss due to the assessed fees. Notably, the only condition for the derivative of $\pi_{liq}$ to be zero is: $(1-\gamma)(1+\ell)-1=0$.  This means that if $\gamma$ is too high compared to $\ell$, liquidators will never choose to enter the system and bad debt will proliferate.
\end{remark}
\subsection{Full Liquidation Event}\label{sec:optimiz-full}
The closing factor and maximum liquidation rate are parameters set by the lending platform. In order to determine the ultimate optimal profit, the optimal profit $L(CF_i,\kappa_i)$ across all potential threshold pairs $(CF_i, \kappa_i)$ should be evaluated, where $i\in N$ indexes the chosen discretizations of thresholds. That is, the maximum would be taken as $\max_{i\in N}\{L(CF_i, \kappa_i)\}$.
In fact, Aave (and major other LPs) use only two threshold pairs: $(CF,1)$ and $(1,\kappa)$ for $CF,\kappa \in (0,1)$. Thus, for simplicity, herein we consider the maximum of the two outcomes, i.e., $\max\{L(CF, 1),\, L(1, \kappa)\}$ as the realized profit.

Notably, there is a key trade-off between the two strategies. If the integral is stopped at a lower health factor, e.g., when the position crosses $CF$, the liquidator is permitted to carry out a larger, final liquidation event in a single shot. In contrast, extending the liquidation sequence by continuing the integral under the $\kappa$ constraint results in a longer stream of smaller profits. The final liquidation's size is usually capped by $\kappa$, which determines the forced cutoff of the permitted liquidation amount. In other words, the liquidator must decide between a shorter horizon with a larger terminal gain or longer horizon with a capped end-payoff. Throughout the remainder of this section, we will explore these tradeoffs via numerical examples.

First, in Example~\ref{ex:strategies}, we study how price impact and market depth affect liquidation outcomes, demonstrating that deeper pools favor the $L(CF,1)$ liquidation strategy. 
Second, we examine a fee-free setting in Example~\ref{ex:k_mults}, to isolate the role of price dynamics, demonstrating how changes in the DEX price cause regime shifts in the binding liquidation constraint and produce nonlinear profit behavior as the dominant liquidation mechanism shifts. Third, in Example~\ref{ex:nofee_nodelta}, we analyze how liquidation profitability changes across health factor regimes by looking at a fixed debt and collateral configuration and changing the underlying asset price.

\begin{example}\label{ex:strategies}
Consider a CPMM with initial reserves of risky assets $A' = 1{,}000s$ and stablecoins $B' = 2{,}000{,}000s$ for $s > 0$ which charges a fee of $\gamma = 0.3\%$. For simplicity of notation, we take $A_0 = 1{,}000$ and $B_0 = 2{,}000{,}000$ for the remainder of this example; the factor $s > 0$ allows us to scale the DEX depth while maintaining the spot price at $B_0/A_0 = 2{,}000$. Fix the lending position at $b = 10{,}000$ and the liquidation parameters to $\theta = 85\%$, $\ell = 5\%$, $CF = 95\%$, and $\kappa = 50\%$. In order to ensure that liquidations can be triggered, the initial collateral level is selected as $c = \frac{bA_0}{\theta B_0}-0.35$.

Under such a setting, we calculate the liquidation profits under two strategies for each value of $s$: $L(CF,1)$, which represents liquidation that is terminated when the health factor recovers to the closing factor and then completed with a final liquidation, and $L(1,\kappa)$, which permits liquidation to proceed up to the protocol-imposed fraction. 

The $L(CF,1)$ strategy produces higher profit for the majority of pool depths, as shown in Figure~\ref{fig:l_compare}, and its advantage grows as the pool gets deeper. 
Shallow pools, causing high slippage, result in liquidations that rapidly hit one of the binding constraints ($x_c$, $x_{\kappa b}$, or $x_{cf}$); because of this, the two strategies---$L(CF,1)$ and $L(1,\kappa)$---behave similarly and sometimes reverse their ordering.
We highlight that $L(CF,1)$ is non-monotonic in pool depth; this occurs as deeper pools reduce $\pi_{liq}$ as $x_{cf}$ is decreasing in pool depth (the decreasing price impacts reduce the length during which the infinitesimal liquidations can occur) while $\pi_{last}$ is increasing.
This explains the divergence seen in the figure and shows how market depth controls the relative optimality of liquidation strategies.

\begin{figure}[!tb]
    \centering
    \includegraphics[scale=0.6]{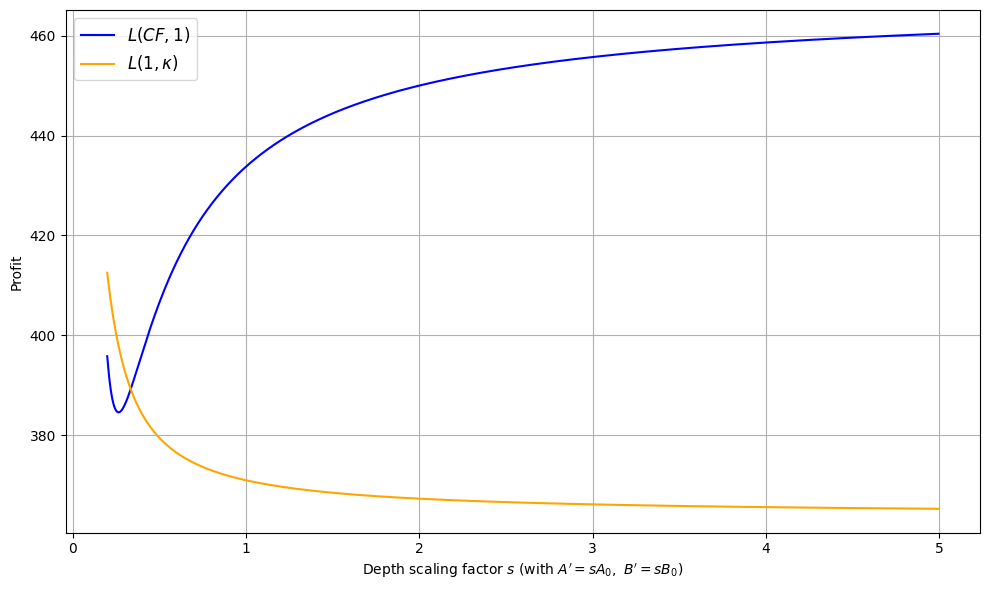}
    \caption{Example~\ref{ex:strategies}: Comparison of liquidation profit under two liquidation strategies.  The x-axis shows the depth of the pool and the y-axis indicates the total profit of the strategy.}
    \label{fig:l_compare}
\end{figure}
\end{example} 

\begin{example}\label{ex:k_mults}
To focus solely on the liquidation problem without considering details of the DEX, we consider the CPMM without any transaction fees, i.e., $\gamma = 0$. In doing this, our goal is to example the relationship between liquidation profits and the asset prices as the initial health factor is varied.

Consider the CPMM with spot price $p > 0$ and liquidity $A_0B_0 = K = 2 \times 10^9$, i.e., $A_0 = \sqrt{K/p}$ and $B_0 = \sqrt{Kp}$. Consider, the lending pool with fixed parameters $\theta = 85\%$, $\ell = 5\%$, $CF = 80\%$ and $\kappa = 50\%$. For an initial health factor ranging from $HF \in \{50\%,90\%,92\%,94\%,99\%\}$, we consider $b = 10,000$ and $c = HF\cdot\frac{bA_0}{\theta B_0}$. Herein we have found that the system behavior is qualitatively comparable for all health factors below $90\%$, as such we consider $HF = 50\%$ as a representative setting.
\begin{figure}[!tb]
    \centering
    \includegraphics[scale=0.6]{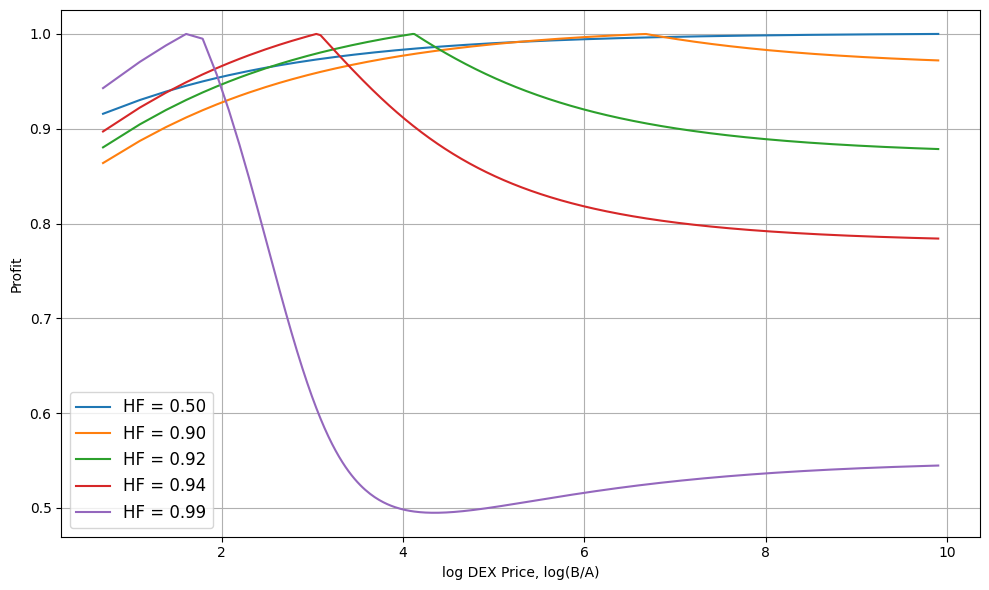}
    \caption{Example~\ref{ex:k_mults}: Liquidation profit for different health factor levels}
    \label{fig:k_mults}
\end{figure}
As seen in Figure~\ref{fig:k_mults}, the profit curves exhibit distinct breakpoints when the variable limiting the amount of liquidations changes. Recall that the available collateral in the lending platform, represented by $x_c$, limits the amount that can be liquidated when the DEX price is low. As the price grows, the requirement to return the health factor to 1 (or to the closing factor $CF$ when $HF < CF$) becomes binding for $x_{\text{liq}}$; that is, the collateral becomes adequate to support deeper liquidations. As a result, the profit function's slope changes as it enters this new regime. At even higher prices, there is a second, more subtle transition; this is particularly noticeable in the $HF = 99\%$ case.  
The $\pi_{\text{last}}$ term captures this second transition, as, again, the liquidator moves from running through all available collateral leaving bad debt for the lending pool to closing out the maximum allowable debt. The decreasing profitability as prices increase results from the decreasing first-stage profits ($\pi_{\text{liq}}$). In fact, even though the final liquidation amount $x_{\text{last}}$ is declining, $\pi_{\text{last}}$ can actually start to rise with the larger initial price. The nonlinearity of $\pi_{\text{last}}$ is the cause of this counterintuitive behavior: each unit becomes more profitable as less is being liquidated meaning the increase in per unit profit exceeds the decreasing liquidation size. Also, eventually, $\pi_{\text{last}}$ starts to grow more quickly than $\pi_{\text{liq}}$ decrease. Consequently, the total profit $\pi_{\text{tot}}$ starts to increase once more for high enough prices.
\end{example}

\begin{example} \label{ex:nofee_nodelta}
In contrast to Example~\ref{ex:k_mults}, we now explore what happens when the health factor is driven by the initial price at the CPMM. Consider the lending pool with debt $b = 10,000$ and collateral $c = 6$. The liquidations are driven by parameters $\theta = 85\%$, $\ell = 5\%$, $CF = 80\%$, and $\kappa = 50\%$. Herein, we consider the CPMM with a fixed liquidity $K = 2 \times 10^9$ with $A_0 = \sqrt{K/p}$ and $B_0 = \sqrt{Kp}$ for $p = HF\cdot\frac{b}{\theta c}$. 
As can be seen in Figure~\ref{fig:nofee_nodelta}, the liquidation profit rises as prices do. The health factor declines when prices fall too low, leaving the lending platform with bad debt that is unrecoverable. Until the price hits 1756.76, the binding constraint on liquidation stays at $x_c$. After that, it changes to $x_{cf}$. The reason for the steep decline in profit around $HF = 1$ is that liquidation is prohibited once the health factor hits 1, which results in a zero profit.
\begin{figure}[!tb]
    \centering
    \includegraphics[scale=0.6]{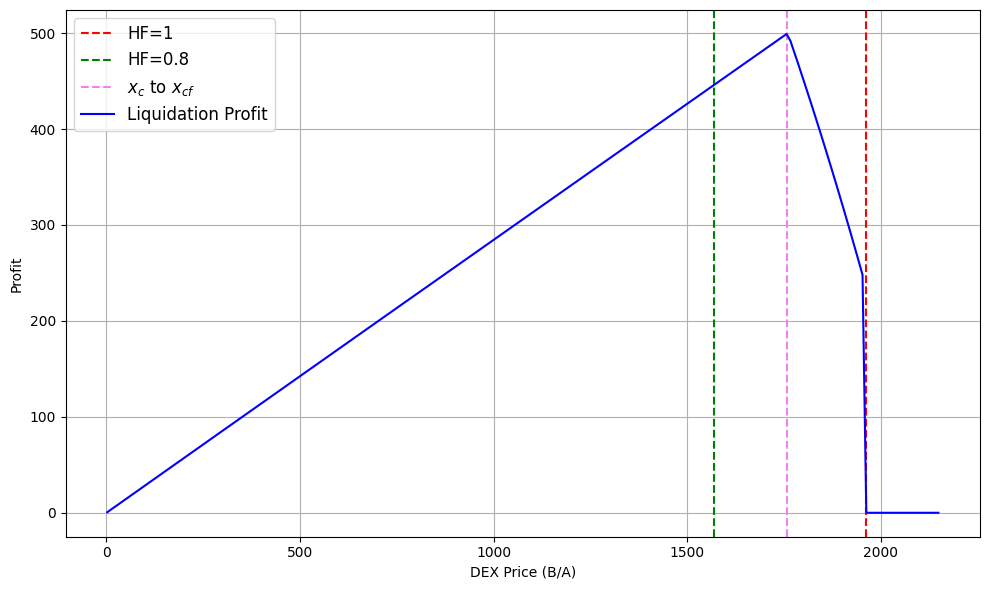}
    \caption{Example~\ref{ex:nofee_nodelta}: Liquidation profit without sandwiching and transaction fees}
    \label{fig:nofee_nodelta}
\end{figure}
\end{example}
\section{Sandwich Attack of a Liquidation Event}\label{sec:front}
Following the common practice in which liquidations are often executed through MEV relays and private transaction bundles, within this section we assume that the liquidator can influence transaction ordering within a block, i.e., get their transaction posted before anyone else tries. This liquidator's goal is to push the collateral price downward to reduce the health factor in order to trigger liquidations, thus, increase profitability. Specifically, within a single block, the liquidator can insert a sequence of trades that: (1) drop the price of the collateral at the DEX; (2) complete the liquidations; and (3) reverse the initial selling pressure at this further depressed price to recover the initial holdings. Notably, such a MEV attack is possible because the liquidator can control their transaction orders and the DEX price can be deterministically predicted within the block.

Even if the initial health factor is above 1, the liquidator can strategically manipulate the price by selling collateral units into the DEX. 
This sale depresses the collateral price and, if sufficiently large, can lower the health factor below the critical threshold to trigger liquidations on the lending platform. 
Mathematically, and as highlighted in the simplified timeline in Figure~\ref{fig:three_steps}, assuming the initial DEX liquidity is $(A_0,B_0)$, the attacker can sell $\Delta$ units of collateral to the DEX to update the DEX reserves ($A_1 > A_0$ and $B_1 < B_0$ in the collateral and debt assets, respectively). 
Once the liquidation spiral unfolds (which causes the DEX to have updated collateral and debt asset reserves $A_2 > A_1$ and $B_2 < B_1$), the liquidator can repurchase the previously sold collateral amount $\Delta$ at the new, lower price. This sandwich attack on the liquidation event results in profits in excess of the liquidation event itself as it effectively allows the liquidator to sell $\Delta$ high and buy it back low. 
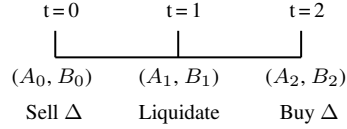
\begin{figure}[!tb]
\centering
\tikzset{every picture/.style={line width=0.75pt}} 
\begin{tikzpicture}[x=0.75pt,y=0.75pt,yscale=-1,xscale=1]
\draw    (124.4,111) -- (248.4,111) ;
\draw    (124.4,97.6) -- (124.4,111) ;
\draw    (248.4,97.6) -- (248.4,111) ;
\draw    (186.4,97.6) -- (186.4,111) ;
\draw (117,82) node [anchor=north west][inner sep=0.75pt]   [align=left] {{\scriptsize t\,=\,0}};
\draw (179,82) node [anchor=north west][inner sep=0.75pt]   [align=left] {{\scriptsize t\,=\,1}};
\draw (241,82) node [anchor=north west][inner sep=0.75pt]   [align=left] {{\scriptsize t\,=\,2}};
\draw (101,115) node [anchor=north west][inner sep=0.75pt]   [align=left] {{\scriptsize ($A_0$, $B_0$)}\\{\scriptsize \,\,\,\,Sell $\Delta $}};
\draw (165,115) node [anchor=north west][inner sep=0.75pt]   [align=left] {{\scriptsize ($A_1$, $B_1$)}\\{\scriptsize Liquidate}};
\draw (229,115) node [anchor=north west][inner sep=0.75pt]   [align=left] {{\scriptsize ($A_2$, $B_2$)}\\{\scriptsize \,\,\,\,Buy $\Delta $}};
\end{tikzpicture}
\caption{\footnotesize{Sell, liquidate and buy back events}}
\label{fig:three_steps}
\end{figure}
\begin{remark}
Within the discussion herein, we ignore where the liquidator obtains the initial $\Delta$ collateral to be sold at the DEX to initiate the sandwich attack. Due to the requirement to close out the sandwich attack after the liquidation event is completed, this initial liquidity could be obtained from, e.g., a flash loan rather than requiring the attacker to have the necessary liquidity on their own.
\end{remark}
Before considering the profits that can be obtained from this sandwich attack, we wish to consider the exact reserves of the CPMM at each stage of the attack. Recall, we assume that the initial CPMM reserves are $(A_0,B_0)$ and the attack is of size $\Delta \geq 0$.
Following the constant product rule, after the initiation of the attack, the DEX reserves update to:
\begin{align*}
A_1(\Delta) &= A_0 + \Delta(1-\gamma) \\
B_1(\Delta) &= \frac{A_0 B_0}{A_0 + \Delta (1-\gamma)}
\end{align*}
so that $A_1(\Delta) B_1(\Delta) = A_0 B_0$.
With these updated reserves, the liquidation process, following the optimal structure discussed in Section~\ref{sec:optimiz}, will commence. Herein, we will simplify the discussion to assume that this process results in $x \geq 0$ of collateral being claimed by the attacking liquidator (not including the $\ell$ bonus paid for completing the liquidation). This liquidation leads to the further, updated, DEX reserves:
\begin{align*}
A_2(\Delta,x) &= A_1(\Delta) + x (1+\ell)(1-\gamma) = A_0 + \Delta (1-\gamma) + x (1+\ell)(1-\gamma)\\
B_2(\Delta,x) &= \frac{A_1(\Delta) B_1(\Delta)}{A_1(\Delta) + x (1+\ell)(1-\gamma)} = \frac{A_0 B_0}{A_0 + \Delta (1-\gamma) + x (1+\ell)(1-\gamma)}
\end{align*}
so that $A_2(\Delta,x) B_2(\Delta,x) = A_1(\Delta) B_1(\Delta) = A_0 B_0$.

In order to determine the true profitability of this full sandwich attack, we need to determine the cost of repurchasing $\Delta$ after the liquidation event has concluded. To ensure that the attacker is repurchasing exactly $\Delta$, and following the CPMM construction, they must pay $\frac{B_2(\Delta,x)\Delta}{(1-\gamma)(A_2(\Delta,x)-\Delta)}$.
In this way, we can write the MEV-based optimization problem:
\begin{align}
\label{eq:MEV-opt} &\max_{\Delta\geq 0}\bigg[\frac{B_0\Delta(1-\gamma)}{A_0+\Delta(1-\gamma)} + V(c,b,A_1(\Delta),B_1(\Delta)) - \frac{B_2(\Delta,x)\Delta}{(1-\gamma)(A_2(\Delta,x)-\Delta)}\bigg]
\end{align}
where $V$ is as in~\eqref{eq:V} while the difference between the first and last terms correspond to the profits purely from the sandwich attack on the liquidation event. Recall that $x$ denotes the optimal amount of collateral claimed (not including the $\ell$ bonus) as part of the liquidation procedure; in this way, implicitly, $x$ is also a function of $\Delta$ within this optimal MEV attack problem.
Notably, this maximization takes into consideration the effect on future liquidation opportunities when determining the ideal amount $\Delta$ to trade in the DEX for MEV profits.
\begin{remark}
While~\eqref{eq:MEV-opt} is written to optimize over $\Delta \geq 0$, we may wish to consider upper and lower bounds on $\Delta$.
First, the MEV attacker implicitly needs to make the initial liquidation large enough to push the health factor of the loan below 1 to trigger any possibility of profits, i.e.,
\[\Delta \geq \max\left\{0 \; , \;  \frac{\sqrt{\frac{\theta c A_0 B_0}{b}}-A_0}{1-\gamma}\right\}.\]
Second, for system robustness, we may impose a constraint so that the DEX holds enough reserves $B_1(\Delta)$ to cover the debt $b$ at the LP, i.e.,
\[\Delta\leq \frac{A_0B_0}{b(1-\gamma)^2(1+\ell)}-\frac{A_0}{1-\gamma}.\]
If $\Delta$ exceeds this cap, then any debt that cannot be liquidated would count as \emph{bad debt} which results in losses for the lenders.\footnote{Implicitly, we are assuming the CPMM is the only market at which the collateral can be sold.}
\end{remark}

\begin{example}\label{ex:front_nofee}
Herein we wish to explore a simple example in which the LP has an initial health factor greater than 1. In this way, without the MEV attack, no liquidation event would be triggered. Thus, to make a profit, the liquidator will need to lower the market price of the collateral by selling some of it into the CPMM in order to start the liquidation spiral. All system parameters are as in Example~\ref{ex:nofee_nodelta}.

Rather than directly solving~\eqref{eq:MEV-opt}, in Figure~\ref{fig:front_nofee}, we plot the profits earned by the liquidators as a function of the attack size $\Delta$.
Notably, the profits are initially zero when $\Delta$ is low because the health factor remains above 1 and no liquidations are triggered. However, as the attack size grows, profits jump up at the moment that the attack drives the health factor to 1 as seen at the dashed red line in the plot. From this point, the total profits earned by the attacker are monotonically increasing as a function of the attack size $\Delta$. Though the total profits continue to grow, the profitability of the liquidation itself begins to drop as a function of the attack size as the liquidation amounts moves from $x_{cf}$ and $x_{\text{last}}$ to $x_c$, i.e., where all collateral is liquidated so the lower prices merely decrease the profits earned from selling that collateral.

We wish to note that the non-monotonicity of the liquidation profits correspond to what was observed in Figure~\ref{fig:nofee_nodelta} (Example~\ref{ex:nofee_nodelta}). In that example, we considered how liquidation profits were a function of the initial price; herein, that initial price is a decreasing function of $\Delta$. Therefore, the same profit dynamics as in Figure~\ref{fig:nofee_nodelta} can be seen if we read Figure~\ref{fig:front_nofee} from right to left.
\begin{figure}[!tb]
    \centering
    \includegraphics[scale=0.6]{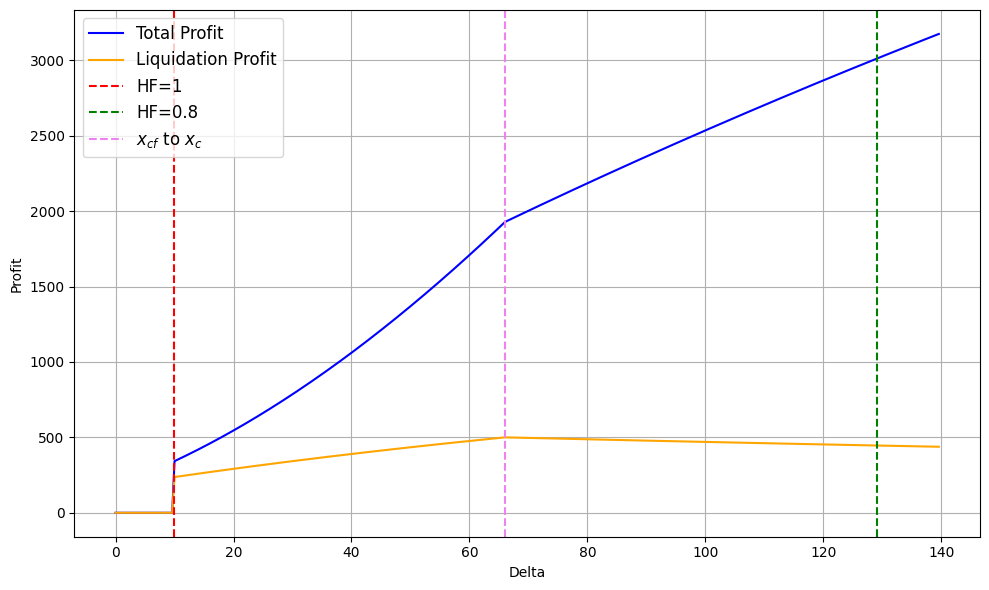}
    \caption{Example~\ref{ex:front_nofee}: Profit with sandwiching and no transaction fees}
    \label{fig:front_nofee}
\end{figure}
\end{example}
\subsection{Limiting Behavior}\label{sec:front-limit}
As seen in Example~\ref{ex:front_nofee}, the liquidator's profits are strictly increasing in the size of the attack. With this intuition, it is tempting to think that the optimal behavior~\eqref{eq:MEV-opt} is always to liquidate as much as possible. 
Within this section, we will investigate the result of that maximally sized attack $\Delta_{max}$ to demonstrate that this is not only wrong, but the worst possible situation when the DEX has a non-zero fee rate $\gamma > 0$.

First, in order to determine the profitability under the maximally sized attack $\Delta_{max}$ we need to determine its size. Specifically, when attempting to close the attack (e.g., to repay the flash loan used to initiate the sandwich attack so as to not revert the transaction), the DEX needs to have sufficient liquidity in order to allow for this purchase, i.e., $A_2(\Delta,x) \geq \Delta$. Because of the fees $\gamma$ charged at the DEX, in Proposition~\ref{pr:delta_max} below, we find the maximal feasible attack size. Note that the attack size $\Delta = A_2(\Delta,x)$ depends on the realized liquidation event $x$; because we are following the situation in which the sandwich attack is of significant size, we simplify to consider $x = x_c$, i.e., all collateral is liquidated in the event.

\begin{proposition}\label{pr:delta_max}
An attack $\Delta \geq 0$ is feasible without reversion (i.e., $A_2(\Delta,x_c) \geq \Delta$) if and only if $\Delta \leq \Delta_{max}$ where
\[\Delta_{max} = \begin{cases} +\infty & \text{if } \gamma = 0 \\ \frac{A_0+(1-\gamma)c}{\gamma}, & \text{if } \gamma > 0. \end{cases}\]
\end{proposition}

\begin{proof}
The first swap in the attack increases the DEX's collateral by $(1-\gamma)\Delta$ units. In the large attack regime under consideration, the liquidation contributes an extra $(1-\gamma)x_c(1+\ell)=(1-\gamma)c$ units of collateral to the DEX after exhausting the borrower's collateral constraint, i.e., $x=x_c=\frac{c}{1+\ell}$. Hence the post-liquidation collateral reserve is $$A_2(\Delta,x_c)=A_0+(1-\gamma)\Delta+(1-\gamma)c\,.$$
The ability of the attacker to repurchase $\Delta$ units of collateral, i.e., $A_2(\Delta,x_c)\ge \Delta$ or equivalently
$$A_0+(1-\gamma)\Delta+(1-\gamma)c \;\ge\; \Delta \quad\leftrightarrow\quad A_0+(1-\gamma)c \;\ge\; \gamma \Delta
$$
is necessary for the attack to be feasible.
This results in $\Delta \le \frac{A_0+(1-\gamma)c}{\gamma}$ if $\gamma>0$. The inequality reduces to $A_0+c\ge 0$ if $\gamma=0$, which holds trivially for $A_0,c\ge 0$, i.e., $\Delta$ is unbounded.
\end{proof}

\begin{remark}
In Proposition~\ref{pr:delta_max}, we implicitly assume $x = x_c = c/(1+\ell)$ from the optimal liquidation after the attack $\Delta$ has commenced. Intuitively, in the large-attack regime in which the price is sufficiently low, the health factor constraint ($x_{cf}$) becomes exceedingly large as recovering the health factor of 1 becomes a near impossibility. Similarly, when the collateral price is low, even ignoring price impacts, the ability to repay the loan $b$ requires a large amount of collateral ($x_b$). This leaves the only remaining constraint to be that based on the total amount of collateral in the LP ($x_c$).
\end{remark}

With this maximal attack size, the goal is to determine the profits that result from such an attack. This profit \emph{and loss} is provided in Proposition~\ref{pr:limit}. Notably, while the case in which the DEX charges no fees ($\gamma = 0$) behaves as seen in Example~\ref{ex:front_nofee} (i.e., larger attacks lead to larger profits), a large attack can generate infinite losses for the attacker when DEX fees are introduced ($\gamma > 0$).

\begin{proposition}\label{pr:limit}
The limiting profitability when $\Delta \nearrow \Delta_{max}$ is provided by:
\begin{itemize}
    \item If $\gamma = 0$: The liquidation value of the collateral $c$, i.e., $\frac{B_0 c}{A_0 + c}$.
    \item If $\gamma > 0$: Infinite losses, i.e., $-\infty$.
\end{itemize}
\end{proposition}

\begin{proof}
In~\eqref{eq:MEV-opt}, we examine the limiting profit as $\Delta \nearrow \Delta_{\max}$, differentiating between $\gamma=0$ and $\gamma>0$.

\noindent\textbf{Case $\gamma=0$.}
When $\gamma=0$, we have $\Delta_{\max}=+\infty$ per Proposition~\ref{pr:delta_max}. The first term in \eqref{eq:MEV-opt} satisfies
$$\lim_{\Delta\to\infty}\frac{B_0\Delta}{A_0+\Delta}=B_0.$$
The continuation value fades since $B_1(\Delta)=\frac{A_0B_0}{A_0+\Delta}\to 0$, implying
$\lim_{\Delta\to\infty}V(c,b,A_1(\Delta),B_1(\Delta))=0.$
The final term converges to
$$\lim_{\Delta\to\infty}\frac{\left(\frac{A_0B_0}{A_0+\Delta+x(1+\ell)}\right)\Delta}{A_0+x(1+\ell)}=\frac{A_0B_0}{A_0+c},$$
where we use $x=x_c=\frac{c}{1+\ell}$. Thus, $\lim_{\Delta\to\infty}\pi(\Delta)=\frac{B_0c}{A_0+c}.$

\noindent\textbf{Case $\gamma>0$.}
When $\gamma>0$, Proposition~\ref{pr:delta_max} implies
\[
\Delta_{\max}=\frac{A_0+(1-\gamma)c}{\gamma}.
\]
As $\Delta\nearrow\Delta_{\max}$, we have $A_2(\Delta,x_c)-\Delta\to 0$, and therefore
\[
\frac{B_2(\Delta,x)\Delta}{(1-\gamma)(A_2(\Delta,x)-\Delta)}\to+\infty.
\]
Since this term enters the objective with a negative sign (while all other terms are bounded), the total profit diverges to $-\infty$.
\end{proof}

\subsection{Fees as Oracle Protection}\label{sec:front-fees}
It follows from Proposition~\ref{pr:limit} that the attacker may need to restrict the attack's scope rather than always carrying out a naive, large attack. Within this section, we consider a numerical example to demonstrate how fees of the DEX oracle, can prevent profitable manipulations to the spot price. 

\begin{example}\label{ex:front-fees}
Consider a joint CPMM and LP system. Herein, we take the CPMM with initial reserves $A_0 = 10{,}000$ and $B_0 = 28{,}000{,}000$ in the collateral asset and stablecoin, respectively. In addition, to match the fees charged at Uniswap v2, we consider $\gamma = 0.3\%$. Consider the LP with $\theta = 85\%$, $\ell = 5\%$, $CF = 80\%$ and $\kappa = 50\%$. Finally, consider the situation in which $b = 32{,}000$ stablecoins were borrowed with $c = 20.12$ of collateral. These values were selected to reflect realistic pool depths and borrower positions.

The profitability of an attack, as a function of the attack size $\Delta$, is provided in Figure~\ref{fig:fee30_delta}. Similar to the earlier examples, we consider both the total profits and loss as well as the profits that are derived directly from the liquidation event. Notably, these transaction fees clearly change the strategy's profitability; for instance, without fees ($\gamma = 0$), the total profits would be non-negative and non-decreasing as in, e.g., Figure~\ref{fig:front_nofee}. Importantly, the overall profit is negative for all positive values of $\Delta > 0$ even though the liquidation profits still reflect the same regime changes discussed in Example~\ref{ex:front_nofee}; that is, once the health factor falls below 1, the liquidation component jumps to a positive value. This effect, in which the attacker takes losses even with the jump in profits at the health factor of 1, deters any such attack.
\begin{figure}[!tb]
    \centering
    \includegraphics[scale=0.6]{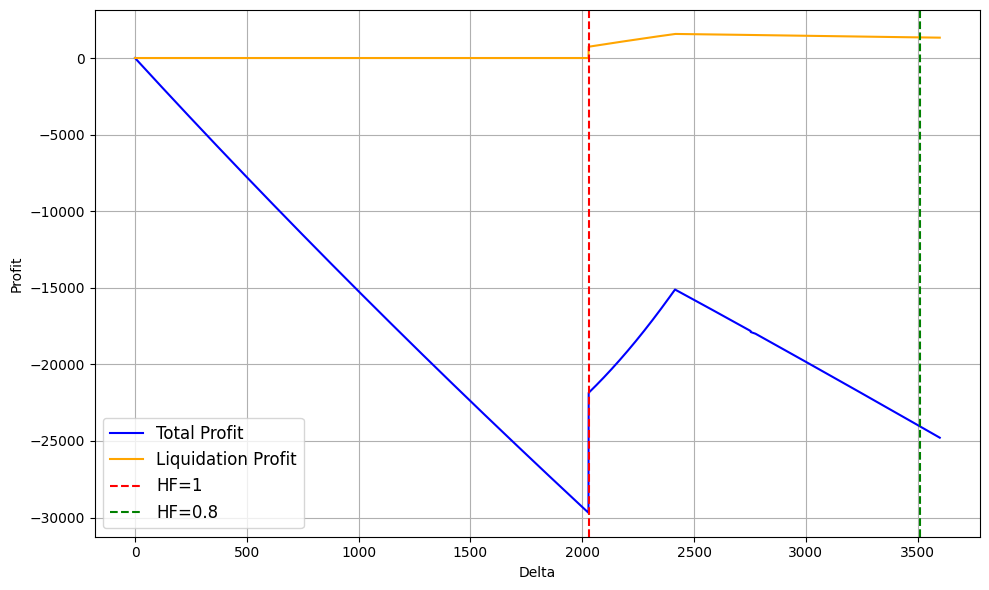}
    \caption{Example~\ref{ex:front-fees}: Profit with 30 bps transaction fees and sandwiching}
    \label{fig:fee30_delta}
\end{figure}
Consider the case in which we vary the CPMM fees $\gamma$. As expected, and shown in Figure~\ref{fig:fee_star}, the total profits extracted by the attack are decreasing in the fees charged. Though a small fee (e.g., $\gamma = 0.1\%$) may still present a profitable sandwiching opportunity for $\Delta$ within a bounded range, raising this fee to $\gamma > \gamma^* := 0.17\%$ completely eliminates the profitability. This demonstrates how transaction fees work as a deterrent to sandwiching because those costs eliminate the economic incentives of the DEX oracle manipulation. That is, the DEX fees serve as a mechanism which promote market stability.
\begin{figure}[!tb]
    \centering
    \includegraphics[scale=0.6]{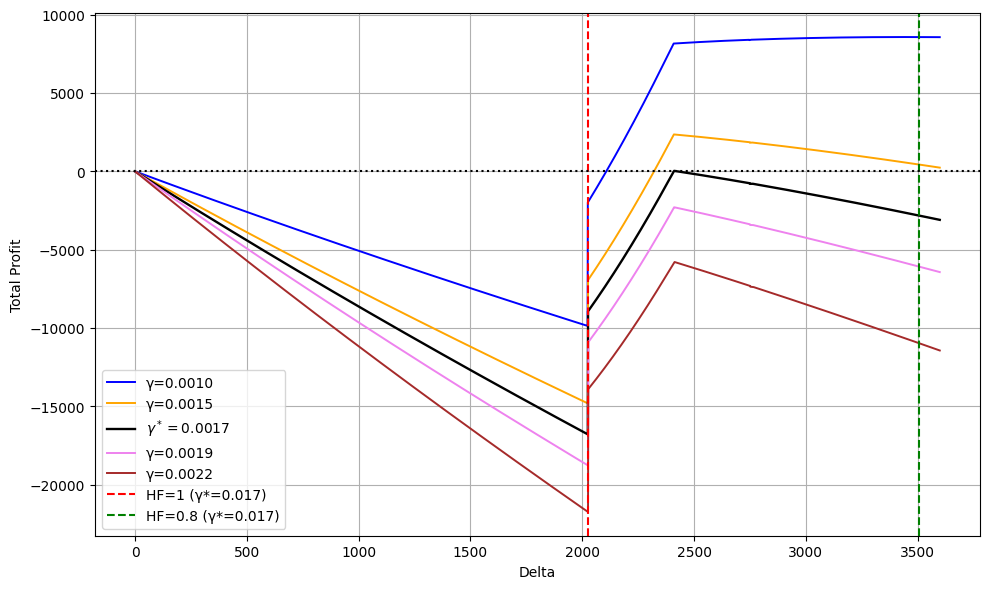}
    \caption{Example~\ref{ex:front-fees}: Total profit as a function of liquidation size $\Delta$ for different fee levels $\gamma$. The black curve corresponds to the critical threshold $\gamma^*=17$ bps.}
    \label{fig:fee_star}
\end{figure}
\end{example}
\section{Discussion}\label{sec:conc}
In order to prevent short-term price manipulation, lending protocols frequently use hardened oracle mechanisms, such as time-weighted averages (TWAPs), medianized prices, or off-chain price feeds. These methods add latency even though they strengthen resistance against transient attacks. Oracle updates that are delayed during sharp market declines may underestimate risk and delay required liquidations, which raises the possibility of bad debt accumulation. In these systems, responsiveness is sacrificed for robustness.

Our findings show that this trade-off is not inevitable. By removing the latency that comes with averaging or off-chain methods, we demonstrate that AMM spot prices can function as practical, fully on-chain oracles. Lending protocols can achieve immediate responsiveness to market movements while preserving solvency during steep price declines by exclusively depending on intra-block DEX prices. This method maintains the blockchain environment's transparency while eliminating reliance on outside data sources.

Importantly, we prove that transaction fees enable this design in two ways. Fees essentially change the incentives for manipulation in addition to rewarding liquidity providers and increasing market liquidity. Unlike previous works (e.g., \cite{cohen23}) that assume frictionless AMMs, we demonstrate that the profitability of oracle manipulation attacks can be completely eliminated by setting DEX fees high enough. Fees act as a ``win-win" design lever in this way, rewarding liquidity provision while fortifying the oracle against predatory MEV tactics. Robustness arises naturally from appropriately aligned market incentives rather than needing external protections.

Crucially, our results serve as a cautious standard for oracle security. The model assumes access to free flash loans and abstracts away operational frictions like gas prices. We give an upper bound on the profitability of a manipulative attacker by disregarding these costs. The robustness of fully on-chain AMM oracles is probably stronger than our theoretical estimates indicate, as the inclusion of gas fees and execution frictions would further compress manipulation incentives in practice.
\bibliographystyle{apalike}
\addcontentsline{toc}{chapter}{References}
{\footnotesize
\bibliography{bibb}}

@TechReport{cornelli24,
    author={Giulio Cornelli and Leonardo Gambacorta and Rodney Garratt and Alessio Reghezza},
    title={{Why DeFi lending? Evidence from Aave V2}},
    year=2024,
    month=May,
    institution={Bank for International Settlements},
    type={BIS Working Papers},
    url={https://ideas.repec.org/p/bis/biswps/1183.html},
    number={1183},
}

@article{zach17,
    title = {Financial contagion and asset liquidation strategies},
    journal = {Operations Research Letters},
    volume = {45},
    number = {2},
    pages = {109-114},
    year = {2017},
    issn = {0167-6377},
    doi = {https://doi.org/10.1016/j.orl.2017.01.004},
    url = {https://www.sciencedirect.com/science/article/pii/S0167637717300329},
    author = {Zachary Feinstein}
}

@inproceedings{qin21,
    author = {Qin, Kaihua and Zhou, Liyi and Livshits, Benjamin and Gervais, Arthur},
    title = {Attacking the DeFi Ecosystem with Flash Loans for Fun and Profit},
    year = {2021},
    isbn = {978-3-662-64321-1},
    publisher = {Springer-Verlag},
    address = {Berlin, Heidelberg},
    url = {https://doi.org/10.1007/978-3-662-64322-8_1},
    doi = {10.1007/978-3-662-64322-8_1},
    pages = {3–32},
    numpages = {30}
}

@article{clark20,
    author = {Clark, Jeremy and Demirag, Didem and Moosavi, Seyedehmahsa},
    title = {Demystifying Stablecoins: Cryptography meets monetary policy},
    year = {2020},
    issue_date = {January-February 2020},
    publisher = {Association for Computing Machinery},
    address = {New York, NY, USA},
    volume = {18},
    number = {1},
    issn = {1542-7730},
    url = {https://doi.org/10.1145/3387945.3388781},
    doi = {10.1145/3387945.3388781},
    journal = {Queue},
    month = mar,
    pages = {39–60},
    numpages = {22}
}

@inproceedings{qin21_1,
   series={IMC ’21},
   title={An empirical study of DeFi liquidations: incentives, risks, and instabilities},
   url={http://dx.doi.org/10.1145/3487552.3487811},
   DOI={10.1145/3487552.3487811},
   booktitle={Proceedings of the 21st ACM Internet Measurement Conference},
   publisher={ACM},
   author={Qin, Kaihua and Zhou, Liyi and Gamito, Pablo and Jovanovic, Philipp and Gervais, Arthur},
   year={2021},
   month=nov, pages={336–350},
   collection={IMC ’21} 
}

@misc{qin23,
      author = {Kaihua Qin and Jens Ernstberger and Liyi Zhou and Philipp Jovanovic and Arthur Gervais},
      title = {Mitigating Decentralized Finance Liquidations with Reversible Call Options},
      howpublished = {Cryptology {ePrint} Archive, Paper 2023/254},
      year = {2023},
      url = {https://eprint.iacr.org/2023/254}
}

@inbook{perez21,
   title={Liquidations: DeFi on a Knife-Edge},
   ISBN={9783662643310},
   ISSN={1611-3349},
   url={http://dx.doi.org/10.1007/978-3-662-64331-0_24},
   DOI={10.1007/978-3-662-64331-0_24},
   booktitle={Financial Cryptography and Data Security},
   publisher={Springer Berlin Heidelberg},
   author={Perez, Daniel and Werner, Sam M. and Xu, Jiahua and Livshits, Benjamin},
   year={2021},
   pages={457–476}
}

@misc{bartoletti22,
      title={Formal Analysis of Lending Pools in Decentralized Finance}, 
      author={Massimo Bartoletti and James Chiang and Tommi Junttila and Alberto Lluch Lafuente and Massimiliano Mirelli and Andrea Vandin},
      year={2022},
      eprint={2206.01333},
      archivePrefix={arXiv},
      primaryClass={cs.SE}
}

@InProceedings{bartoletti21,
    author="Bartoletti, Massimo
    and Chiang, James Hsin-yu
    and Lafuente, Alberto Lluch",
    title="SoK: Lending Pools in Decentralized Finance",
    booktitle="Financial Cryptography and Data Security. FC 2021 International Workshops",
    year="2021",
    publisher="Springer Berlin Heidelberg",
    address="Berlin, Heidelberg",
    pages="553--578",
    isbn="978-3-662-63958-0"
}

@INPROCEEDINGS{gudgeon20_1,
  author={Gudgeon, Lewis and Perez, Daniel and Harz, Dominik and Livshits, Benjamin and Gervais, Arthur},
  booktitle={2020 Crypto Valley Conference on Blockchain Technology (CVCBT)}, 
  title={The Decentralized Financial Crisis}, 
  year={2020},
  volume={},
  number={},
  pages={1-15},
  doi={10.1109/CVCBT50464.2020.00005}
}

@inproceedings{gudgeon20,
    author = {Gudgeon, Lewis and Werner, Sam and Perez, Daniel and Knottenbelt, William J.},
    title = {DeFi Protocols for Loanable Funds: Interest Rates, Liquidity and Market Efficiency},
    year = {2020},
    isbn = {9781450381390},
    publisher = {Association for Computing Machinery},
    address = {New York, NY, USA},
    url = {https://doi.org/10.1145/3419614.3423254},
    doi = {10.1145/3419614.3423254},
    booktitle = {Proceedings of the 2nd ACM Conference on Advances in Financial Technologies},
    pages = {92–112},
    numpages = {21},
    location = {New York, NY, USA},
    series = {AFT '20}
}

@inproceedings{klages20,
    author = {Klages-Mundt, Ariah and Harz, Dominik and Gudgeon, Lewis and Liu, Jun-You and Minca, Andreea},
    title = {Stablecoins 2.0: Economic Foundations and Risk-based Models},
    year = {2020},
    isbn = {9781450381390},
    publisher = {Association for Computing Machinery},
    address = {New York, NY, USA},
    doi = {10.1145/3419614.3423261},
    booktitle = {Proceedings of the 2nd ACM Conference on Advances in Financial Technologies},
    pages = {59–79},
    numpages = {21},
    location = {New York, NY, USA},
    series = {AFT '20}
}

@misc{iragorri21,
      title={Financial intermediation and risk in decentralized lending protocols}, 
      author={Carlos Castro-Iragorri and Julian Ramirez and Sebastian Velez},
      year={2021},
      eprint={2107.14678},
      archivePrefix={arXiv},
      primaryClass={q-fin.GN}
}

@article{john23,
   author = "John, Kose and Kogan, Leonid and Saleh, Fahad",
   title = "Smart Contracts and Decentralized Finance", 
   journal= "Annual Review of Financial Economics",
   year = "2023",
   volume = "15",
   number = "Volume 15, 2023",
   pages = "523-542",
   doi = "https://doi.org/10.1146/annurev-financial-110921-022806",
   publisher = "Annual Reviews",
   issn = "1941-1375",
   type = "Journal Article"
}

@article{cohen23,
    author = {Samuel N Cohen and Marc Sabate-Vidales and Lukasz Szpruch and Mathis Gontier Delaunay},
    title = {The Paradox of Adversarial Liquidation in Decentralised Lending},
    journal = {SSRN Electronic Journal},
    year = {2023},
    publisher = {Social Science Electronic Publishing},
    month = {aug},
    url = {https://doi.org/10.2139/ssrn.4540333},
    doi = {10.2139/ssrn.4540333}
}

@article{lehar22,
  title={Systemic Fragility in Decentralized Markets},
  author={Alfred Lehar and Christine A. Parlour},
  journal={SSRN Electronic Journal},
  year={2022},
  url={https://api.semanticscholar.org/CorpusID:250417631}
}

@misc{szpruch24,
      title={Pricing and hedging of decentralised lending contracts}, 
      author={Lukasz Szpruch and Marc Sabaté Vidales and Tanut Treetanthiploet and Yufei Zhang},
      year={2024},
      eprint={2409.04233},
      archivePrefix={arXiv},
      primaryClass={q-fin.PR},
      url={https://arxiv.org/abs/2409.04233}, 
}

@article{mueller24,
  title = {DeFi Leveraged Trading: Inequitable Costs of Decentralization},
  author = {Mueller, Peter},
  year = {2024}
}

@INPROCEEDINGS{daian20,
  author={Daian, Philip and Goldfeder, Steven and Kell, Tyler and Li, Yunqi and Zhao, Xueyuan and Bentov, Iddo and Breidenbach, Lorenz and Juels, Ari},
  booktitle={2020 IEEE Symposium on Security and Privacy (SP)}, 
  title={Flash Boys 2.0: Frontrunning in Decentralized Exchanges, Miner Extractable Value, and Consensus Instability}, 
  year={2020},
  volume={},
  number={},
  pages={910-927},
  doi={10.1109/SP40000.2020.00040}
}

@inproceedings{li23,
   title={Demystifying DeFi MEV Activities in Flashbots Bundle},
   url={http://dx.doi.org/10.1145/3576915.3616590},
   DOI={10.1145/3576915.3616590},
   booktitle={Proceedings of the 2023 ACM SIGSAC Conference on Computer and Communications Security},
   publisher={ACM},
   author={Li, Zihao and Li, Jianfeng and He, Zheyuan and Luo, Xiapu and Wang, Ting and Ni, Xiaoze and Yang, Wenwu and Chen, Xi and Chen, Ting},
   year={2023},
   month=nov, pages={165–179},
   collection={CCS ’23}
}

@INPROCEEDINGS{qin22,
  author={Qin, Kaihua and Zhou, Liyi and Gervais, Arthur},
  booktitle={2022 IEEE Symposium on Security and Privacy (SP)}, 
  title={Quantifying Blockchain Extractable Value: How dark is the forest?}, 
  year={2022},
  volume={},
  number={},
  pages={198-214},
  doi={10.1109/SP46214.2022.9833734}
}

@InProceedings{wang22,
    author="Wang, Zhipeng
    and Qin, Kaihua
    and Minh, Duc Vu
    and Gervais, Arthur",
    editor="Eyal, Ittay
    and Garay, Juan",
    title="Speculative Multipliers on DeFi: Quantifying On-Chain Leverage Risks",
    booktitle="Financial Cryptography and Data Security",
    year="2022",
    publisher="Springer International Publishing",
    address="Cham",
    pages="38--56"
}

@TechReport{auer22,
  author={Raphael Auer and Jon Frost and Jose María Vidal Pastor},
  title={{Miners as intermediaries: extractable value and market manipulation in crypto and DeFi}},
  year=2022,
  month=Jun,
  institution={Bank for International Settlements},
  type={BIS Bulletins},
  number={58}
}

@misc{eskandari19,
      title={SoK: Transparent Dishonesty: front-running attacks on Blockchain}, 
      author={Shayan Eskandari and Seyedehmahsa Moosavi and Jeremy Clark},
      year={2019},
      eprint={1902.05164},
      archivePrefix={arXiv}
}

@INPROCEEDINGS{zhou21,
  author={Zhou, Liyi and Qin, Kaihua and Torres, Christof Ferreira and Le, Duc V and Gervais, Arthur},
  booktitle={2021 IEEE Symposium on Security and Privacy (SP)}, 
  title={High-Frequency Trading on Decentralized On-Chain Exchanges}, 
  year={2021},
  pages={428-445},
  doi={10.1109/SP40001.2021.00027}
}

@misc{yang23,
      title={SoK: MEV Countermeasures: Theory and Practice}, 
      author={Sen Yang and Fan Zhang and Ken Huang and Xi Chen and Youwei Yang and Feng Zhu},
      year={2023},
      eprint={2212.05111},
      archivePrefix={arXiv},
      primaryClass={cs.CR}
}

@inproceedings{heimach22,
   title={Eliminating Sandwich Attacks with the Help of Game Theory},
   url={http://dx.doi.org/10.1145/3488932.3517390},
   DOI={10.1145/3488932.3517390},
   booktitle={Proceedings of the 2022 ACM on Asia Conference on Computer and Communications Security},
   publisher={ACM},
   author={Heimbach, Lioba and Wattenhofer, Roger},
   year={2022},
   month=may, pages={153–167},
   collection={ASIA CCS ’22}
}

@misc{zhou21-1,
      title={A2MM: Mitigating Frontrunning, Transaction Reordering and Consensus Instability in Decentralized Exchanges}, 
      author={Liyi Zhou and Kaihua Qin and Arthur Gervais},
      year={2021},
      eprint={2106.07371},
      archivePrefix={arXiv},
      primaryClass={cs.CR},
      url={https://arxiv.org/abs/2106.07371}, 
}

@techreport{capponi25,
  title        = {Maximal Extractable Value and Allocative Inefficiencies in Public Blockchains},
  author       = {Capponi, Agostino and Jia, Ruizhe and Wang, Ye},
  institution  = {SSRN},
  type         = {SSRN Scholarly Paper},
  number       = {3997796},
  year         = {2025},
  month        = {February},
  url          = {https://ssrn.com/abstract=3997796},
  doi          = {10.2139/ssrn.3997796}
}

@misc{lee23,
      title={All AMMs are CFMMs. All DeFi markets have invariants. A DeFi market is arbitrage-free if and only if it has an increasing invariant}, 
      author={Roger Lee},
      year={2023},
      eprint={2310.09782},
      archivePrefix={arXiv},
      primaryClass={q-fin.TR},
      url={https://arxiv.org/abs/2310.09782}, 
}

@misc{capponi23,
      title={Decentralized Finance: Protocols, Risks, and Governance}, 
      author={Agostino Capponi and Garud Iyengar and Jay Sethuraman},
      year={2023},
      eprint={2312.01018},
      archivePrefix={arXiv},
      primaryClass={q-fin.TR},
      url={https://arxiv.org/abs/2312.01018}, 
}

@misc{iftikhar25,
      title={Automated Risk Management Mechanisms in DeFi Lending Protocols: A Crosschain Comparative Analysis of Aave and Compound}, 
      author={Erum Iftikhar and Wei Wei and John Cartlidge},
      year={2025},
      eprint={2506.12855},
      archivePrefix={arXiv},
      url={https://arxiv.org/abs/2506.12855}, 
}

@inproceedings{angeris20,
    author = {Angeris, Guillermo and Chitra, Tarun},
    title = {Improved Price Oracles: Constant Function Market Makers},
    year = {2020},
    isbn = {9781450381390},
    publisher = {Association for Computing Machinery},
    address = {New York, NY, USA},
    url = {https://doi.org/10.1145/3419614.3423251},
    doi = {10.1145/3419614.3423251},
    booktitle = {Proceedings of the 2nd ACM Conference on Advances in Financial Technologies},
    pages = {80–91},
    numpages = {12},
    location = {New York, NY, USA},
    series = {AFT '20}
}

@misc{angeris21,
      title={Constant Function Market Makers: Multi-Asset Trades via Convex Optimization}, 
      author={Guillermo Angeris and Akshay Agrawal and Alex Evans and Tarun Chitra and Stephen Boyd},
      year={2021},
      eprint={2107.12484},
      archivePrefix={arXiv},
      primaryClass={math.OC},
      url={https://arxiv.org/abs/2107.12484}, 
}

@ARTICLE{cong21,
    title = {Tokenomics: Dynamic Adoption and Valuation},
    author = {Cong, Lin and Li, Ye and Wang, Neng},
    year = {2021},
    journal = {The Review of Financial Studies},
    volume = {34},
    number = {3},
    pages = {1105-1155}
}

@misc{capponi21,
      title={The Adoption of Blockchain-based Decentralized Exchanges}, 
      author={Agostino Capponi and Ruizhe Jia},
      year={2021},
      eprint={2103.08842},
      archivePrefix={arXiv},
      primaryClass={q-fin.TR}
}

@incollection{imf22,
  author       = {{International Monetary Fund}},
  title        = {The Rapid Growth of Fintech: Vulnerabilities and Challenges for Financial Stability},
  booktitle    = {Global Financial Stability Report},
  year         = {2022},
  month        = apr,
  publisher    = {International Monetary Fund},
  address      = {Washington, DC},
  chapter      = {3},
  url          = {https://www.imf.org/-/media/files/publications/gfsr/2022/april/english/ch3.pdf}
}
\newpage
\appendix
\section{Proofs}\label{sec:app}
\subsection{Lemma~\ref{prop:dp} Proof}\label{app:dp}
There are two conditions that should be checked in order for the value function in~\eqref{eq:V} to be the sum of the two profit functions, i.e., step (7) in Algorithm~\ref{algo:pi_liq}.
\begin{itemize}
    \item Profit subadditivity:
    \begin{align*}
        \pi_{liq}(x_1+x_2) \leq \pi_{liq}(x_1) + \pi_{liq}(x_2; x_1)
    \end{align*}
where the notation $x_2;x_1$ means using the value $x_2$ starting from $x_1$.
    \item Health factor constraint:
    \begin{align*}
        & HF(x_1+x_2) \geq HF(x_2; x_1) \\
        & \text{s.t.} \quad HF(x_1) \leq CF
    \end{align*}
As based on the profit subadditivity condition we realized that smaller transactions are better, they also should be able to increase the HF at a slower rate (up to $CF$), so that we will have more opportunities to liquidate. Hence, based on~\eqref{eq:hf_schem}:
\begin{align*}
    HF(x_1+x_2) &= \frac{\theta(c-(x_1+x_2)(1+\ell))\frac{B_1}{A_1}}{b-(x_1+x_2)\frac{B_0}{A_0}}\\
    HF(x_2;x_1) &= \frac{\theta(c-x_1(1+\ell)-x_2(1+\ell))\frac{B_2}{A_2}}{b-x_1\frac{B_0}{A_0}-x_2\frac{B_1}{A_1}}
\end{align*}
As $\frac{B_0}{A_0}\geq\frac{B_1}{A_1}\geq\frac{B_2}{A_2}$, the denominator of $HF(x_2;x_1)$ is larger than $HF(x_1+x_2)$ and the numerator smaller, making the inequality to hold.
\end{itemize}

As highlighted above, the optimal liquidation strategy is to incrementally claim and sell marginal $dx$ units of the collateral. A question that might rise is if liquidating all at once would yield a higher profit. If we liquidate all in one transaction, the profit is as:
\begin{align*}
    \pi_{one}=\frac{Bx(1+\ell)(1-\gamma)}{A+x(1+\ell)(1-\gamma)}-\frac{Bx}{A}
\end{align*}

If we liquidate by incrementally small amounts, each transaction's value is calculated by multiplying the quantity sold by the execution price. The relevant value for the transaction is based on the updated price after selling up to $x$, or at the price using $B - y$, where $y$ is obtained from~\eqref{eq:y}, rather than the initial price $B$ because every sale lowers the price. Consequently, $\frac{B - y}{A+x(1-\gamma)(1+\ell)}$ is used as the price per unit at that point in time to determine the value of the subsequent small transaction $dx$:
\begin{align}
    \frac{(B-y)dx(1+\ell)(1-\gamma)}{A+x(1+\ell)(1-\gamma)+dx(1+\ell)(1-\gamma)}-\frac{(B-y)dx}{A+x(1+\ell)(1-\gamma)}\nonumber\\
    =\frac{(B-y) dx((1+\ell)(1-\gamma)-1)}{A+x(1+\ell)(1-\gamma)}
    \label{eq:dx_diff_fee}
\end{align}
The total profit from these liquidations then would be the integral up to $x_{liq}$:
\begin{align}
    \pi_{liq}=\int_0^{x_{liq}}& \frac{(B-y)((1-\gamma)(1+\ell)-1)}{A+x(1-\gamma)(1+\ell)}dx\nonumber\\
    &= \frac{B((1-\gamma)(1+\ell)-1)x_{liq} }{A+ x_{liq}(1-\gamma)(1+\ell)}
    \label{eq:dx_profit_fee}
\end{align}
where $x_{liq}=\min\{x_c, x_b, x_{cf}\}$ and $x_c$ and $x_b$ can be obtained from~\eqref{eq:xc} and~\eqref{eq:xb} with $\kappa=1$ respectively and $x_{cf}$ is detailed in~\eqref{eq:x_cf}. 

To test which one has a higher profit, we compare the profit from the liquidation types ($\pi_{liq}, \text{vs.}\,\pi_{one}$). Note that, for any $x \geq 0$:
\[\frac{B((1-\gamma)(1+\ell)-1)x }{A+ x(1-\gamma)(1+\ell)} = \frac{B x (1+\ell)(1-\gamma)}{A + x(1+\ell)(1-\gamma)} - \frac{Bx}{A + x(1+\ell)(1-\gamma)} \geq \frac{B x (1+\ell)(1-\gamma)}{A + x(1+\ell)(1-\gamma)} - \frac{B x}{A}\]
This means liquidating by smaller transactions is more profitable.

\end{document}